\pdfoutput=1
\documentclass[journal,10pt]{IEEEtran}
\usepackage[letterpaper, left=0.625in, right=0.625in, top=0.75in, bottom=1.0in]{geometry}

\usepackage{booktabs}
\usepackage{stix}

\usepackage{amsfonts}
\usepackage{xcolor,soul,framed} 

\colorlet{shadecolor}{yellow}
\usepackage[pdftex]{graphicx}
\graphicspath{{../pdf/}{../jpeg/}}
\DeclareGraphicsExtensions{.pdf,.png,.jpg,.jpeg}

\usepackage[cmex10]{amsmath}

\usepackage{array}
\usepackage{mdwmath}
\usepackage{mdwtab}
\usepackage{eqparbox}
\usepackage{url}
\usepackage{blkarray}

\usepackage{amsmath, amsfonts, amssymb}
\usepackage{mathtools}

\usepackage{algorithm}
\usepackage{algpseudocode}
\usepackage{float} 
\usepackage{subcaption}
\usepackage[T1]{fontenc}  

\usepackage{etoolbox}
\makeatletter
\patchcmd{\ALG@doentity}{\noindent}{}{}{}
\makeatother
\usepackage{mathrsfs}

\usepackage{tikz}
\usetikzlibrary{arrows.meta,shapes,positioning}

\usepackage[dvipsnames]{xcolor}
\usepackage[
  colorlinks=false,            
  citebordercolor=ForestGreen, 
  linkbordercolor=red,       
  urlbordercolor=white,        
  pdfborderstyle={/S/S/W 1},   
  bookmarks=false
]{hyperref}

\usepackage{booktabs}
\usepackage{tabularx}
\usepackage{amsmath}

\usepackage{graphicx}
\usepackage{subcaption} 

\newtheorem{definition}{\textbf{Assumption}}
\newtheorem{proposition}{\textbf{Proposition}}
\newtheorem{corollary}{\textbf{Corollary}}

\algnewcommand\Output{\textbf{Output: }}

\begin{document}
 \setlength{\columnsep}{0.21in}
 \bstctlcite{IEEEexample:BSTcontrol}
     \title{Adaptive Client Clustering and Coordination for Federated Learning Workflow Management in Edge Networks}
%


\author{\IEEEauthorblockN{Jieping Luo, Qiyue Li,~\IEEEmembership{Student~Member,~IEEE}, Yuxuan Chen, Hang Qi,~\IEEEmembership{Student~Member,~IEEE}, Jiaying Yin, Jingjin Wu, \emph{Senior~Member, IEEE}, and Qian~Wang,~\IEEEmembership{Senior~Member,~IEEE}}

\thanks{
J. Luo is with the Department of Statistics, University of Oxford, 24–29 St Giles’, Oxford OX1 3LB, U.K. Email: \href{mailto:jieping.luo@reuben.ox.ac.uk}{jieping.luo@reuben.ox.ac.uk}. 

Q. Li, Y. Chen, H. Qi, and J. Wu are with the Guangdong Provincial/Zhuhai Key Laboratory of IRADS, Beijing Normal-Hong Kong Baptist University, Zhuhai, China. Emails: 
\href{mailto:t330005026@mail.bnbu.edu.cn}{t330005026@mail.bnbu.edu.cn}, 
\href{mailto:t330018009@mail.bnbu.edu.cn}{t330018009@mail.bnbu.edu.cn}, 
\href{mailto:t330202706@mail.bnbu.edu.cn}{t330202706@mail.bnbu.edu.cn}, 
\href{mailto:jj.wu@ieee.org}{jj.wu@ieee.org}. 

J. Yin is with the Institute of Precision Medicine, The First Affiliated Hospital, Sun Yat-Sen University, Guangzhou, Guangdong, 510080, P.R. China. Email:\href{mailto:yinjy35@mail.sysu.edu.cn}{yinjy35@mail.sysu.edu.cn}. 

Q. Wang is with the Institute of Cyberspace Security, Zhejiang University of Technology, Hangzhou 310023, China. Email:\href{mailto:wangqian18@zjut.edu.cn}{wangqian18@zjut.edu.cn}.

\textit{Corresponding author: J. Wu.}
}}



\maketitle

%

\begin{abstract}
Federated learning (FL) is increasingly deployed as a managed learning service rather than as a set of isolated training jobs. In networked edge environments, dependent FL service flows must coordinate heterogeneous clients, non-IID data, fluctuating communication latency, and precedence-constrained tasks under service-level completion requirements. These coupled factors make participant management central to both time-to-target performance and learning stability. This paper proposes A-CoDa, an adaptive clustered coordination framework for managing dependent FL flows. A-CoDa first uses label-distribution divergence (LDD)-based greedy-balanced clustering to construct statistically coherent and size-aware client groups, which serve as a scalable management abstraction. Building on this structure, we design FedMIX, an uncertainty-aware intra-/inter-cluster participation mechanism that ranks clients by a loss--latency--uncertainty utility and adaptively controls cross-cluster probing according to training progress and latency conditions. A dependency-aware DAG scheduler then orchestrates layer-wise task execution so that parallelism and precedence constraints are jointly respected. We further provide a convergence analysis that frames the result as a sufficient loss-domain design bound, explicitly relating the attainable error floor and sufficient communication rounds to LDD-induced sampling mismatch, residual distribution shift, local-SGD drift, stochastic variance, and adaptive probing budgets. Experiments on handwriting, wearable-sensing, product-image, and medical-imaging tasks evaluate A-CoDa under dependent FL workflows and demonstrate its effectiveness in reducing end-to-end completion time while maintaining competitive accuracy.
\end{abstract}


%
\IEEEpeerreviewmaketitle


\section{Introduction}
Federated learning (FL)~\cite{yang2019federated} enables distributed clients to collaboratively train models while keeping raw data local. Beyond single-model training, FL is increasingly used as part of managed networked services in which multiple learning tasks must be executed with service-level requirements. Examples include privacy-preserving demand forecasting for supply-chain and retail operations~\cite{qi2025comparative}, healthcare and medical-assistance workflows~\cite{datta2024blockchain}, mobile and wearable sensing analytics~\cite{11223124}, product-inspection services~\cite{10206024}, and networked edge intelligence for satellite or remote-connectivity scenarios~\cite{qi2025energy}. In these applications, the learning service is not only judged by final model accuracy, but also by whether all dependent tasks can finish within a reasonable wall-clock time while respecting privacy, communication, and resource constraints.

Managing such FL service flows is challenging because statistical and system heterogeneity interact. Clients may hold different label distributions, collect task-specific data at different rates, and experience different communication and computation delays. A client with rare or difficult samples may be statistically valuable but slow, while a fast client may contribute redundant information. Therefore, client participation should be treated as an online management decision that jointly considers learning utility, latency cost, and exploration of under-sampled participants, instead of simply selecting the largest or fastest set of clients.

Dependent FL workflows further complicate this decision. In a single FL job, a slow round only delays that job. In a directed acyclic graph (DAG) of FL tasks, however, the same delay may postpone downstream services on the critical execution path. For example, a downstream sensing, inspection, or diagnosis model may not start until an upstream representation, filtering, or pre-processing model has met its target. Static client selection or one-shot task-to-client matching is therefore inefficient when network states, local losses, and client availability evolve over time~\cite{8761315}. Client clustering can reduce the search complexity, but a purely static cluster assignment cannot react to changing training progress or cross-cluster sampling bias.

We propose A-CoDa (Adaptive Cluster-oriented and Dependency-aware Hierarchical Client Selection for Federated Learning), a framework for managing dependent FL service flows under heterogeneous edge resources. A-CoDa combines LDD-based balanced clustering, uncertainty-aware intra-/inter-cluster participation control, and DAG-aware scheduling. The LDD-based clustering stage constructs statistically coherent and size-balanced client groups that serve as a stable low-complexity management abstraction. FedMIX then performs in-cluster exploitation while adaptively probing external clusters when additional diversity is useful, ranking candidates with a loss--latency--uncertainty utility and controlling the external probing budget according to training progress and latency conditions. The DAG scheduler coordinates layer-wise execution so that task precedence and available parallelism are both respected.

Existing FL client-selection schemes are mainly static, dynamic, or cluster-based. Static methods reduce overhead but ignore time-varying system states~\cite{cho2022towards}; dynamic methods exploit latency, computation, or loss feedback~\cite{9443523,9846900,lai2021oort} but usually focus on a single job; and cluster-based methods exploit statistical or resource similarity~\cite{wang2025fedccs,10074237,ren2025dynamic}, but are often static or require repeated re-clustering. Multi-job FL scheduling has been explored~\cite{zhou2022efficient,9685698}, but most existing formulations assume independent jobs rather than precedence-constrained service flows. These limitations motivate a design that keeps the efficiency of clustered coordination while allowing round-level participation management under dependent tasks.

Another limitation is that many selection rules optimize short-term utility, such as expected loss decrease or deadline feasibility, without explicitly tracking how repeated biased participation affects the aggregate update over multiple rounds. In heterogeneous learning services, this can cause a cluster to converge quickly on its dominant distribution while under-representing useful external data. A-CoDa addresses this by preserving an in-cluster backbone for efficiency and adding a bounded, adaptive external probing channel for controlled correction. The uncertainty term in FedMIX discourages stale participation estimates, while the adaptive budget reduces unnecessary probing under latency congestion.

This work extends our prior conference paper~\cite{luo2025cluster}, where client coordination remained largely static after initialization. The inherited components include the basic LDD clustering idea and the PPO-based dependency scheduler. The journal extension focuses on online flow management: it introduces the revised FedMIX participation rule with loss--latency--uncertainty scoring, adaptive inter-cluster probing budgets, a refined convergence analysis that treats the result as a sufficient design bound rather than a direct predictor of empirical accuracy thresholds, and a broader evaluation of dependent FL workflows.

Our main contributions are summarized as follows.
\begin{itemize}
    \item We formulate dependent FL as a managed service-flow coordination problem, where client participation, task dependencies, and latency jointly determine the end-to-end completion time. A-CoDa integrates LDD-based balanced clustering with DAG-aware task execution to provide a scalable coordination architecture.
    \item We design FedMIX, an uncertainty-aware intra-/inter-cluster participation mechanism. FedMIX ranks candidates using a loss--latency--uncertainty utility and adaptively controls cross-cluster probing according to training stagnation, external uncertainty, and latency congestion.
    \item We provide a sufficient convergence and completion-time analysis under PL-type objectives. The analysis explicitly separates sampling mismatch, residual distribution shift, local-SGD drift, stochastic variance, and adaptive probing budgets, clarifying how FedMIX controls the bias--variance--latency trade-off without claiming to exactly predict measured accuracy thresholds.
    \item We evaluate A-CoDa on a dependent multi-task FL suite involving handwriting, wearable sensing, product-image classification, and medical imaging, together with scalability tests up to 500 clients. The results show improved end-to-end completion time while maintaining competitive accuracy.
\end{itemize}

The rest of this paper is organized as follows. Section~II summarizes the related work. Section~III describes the system model. Section~IV provides the convergence and completion-time analysis. Section~V presents the proposed coordination algorithms. Section~VI reports the numerical results. Section~VII concludes the paper.

\section{Related Work}

\subsection{Network and Service Management for Edge Intelligence}
Network and service management traditionally emphasizes the coordinated management of resources, services, policies, reliability, and performance across networked systems. In modern edge-intelligence settings, learning tasks themselves become managed services: they consume communication and computation resources, interact with privacy constraints, and must satisfy time-to-target or service-completion requirements. This perspective is closely related to task-oriented communication and edge-intelligence orchestration, where sensing, communication, computation, and learning objectives are jointly considered~\cite{10195234,9498853,10713971}. Unlike conventional resource-allocation studies that optimize a single offloading or communication objective, A-CoDa treats FL participation as a service-flow management action: the selected clients affect both learning progress and the completion time of downstream tasks.

\subsection{Distribution-Divergence-Based Clustering in Federated Learning}
Distributional discrepancy measures are widely used to quantify non-IID data heterogeneity in FL~\cite{10468591}. Optimal-transport and Wasserstein-style metrics capture structural differences when a ground metric is available~\cite{rubner2000earth,zhang2022deepemd,shen2018wasserstein}, while simpler histogram divergences are often sufficient for label-skewed FL. Prior works used such discrepancies for personalization, divergence-aware aggregation, and clustering~\cite{chen2022emd,9797945,alekseenko2024distance,luo2023influence,10018536}. Following our implementation, this paper uses an $\ell_1$ label-distribution divergence (LDD), rather than solving an optimal-transport problem. This choice keeps the clustering metric lightweight and aligned with empirical label histograms that can be obtained without collecting raw records.

The lightweight nature of LDD is important for managed FL flows. A service manager can often obtain class-count summaries or coarse distribution statistics, whereas richer feature-distribution estimates may require additional inference, communication, or privacy-sensitive metadata. LDD is therefore used as a low-overhead proxy for dominant label-skew heterogeneity, not as a universal substitute for feature-level distribution shift.

\begin{figure}[t]
    \centering
    \includegraphics[width=0.94\linewidth]{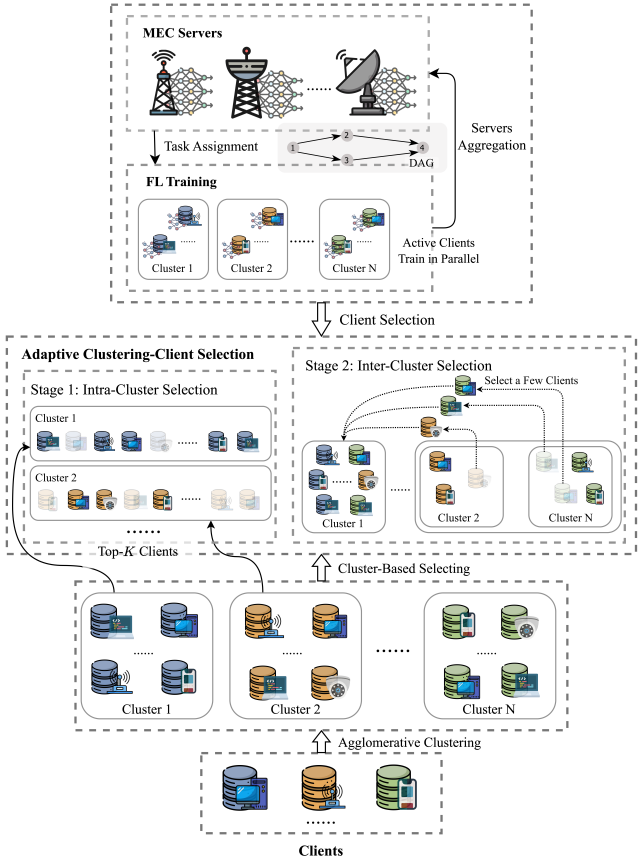}
    \caption{Adaptive cluster-based FL client selection in networked edge systems.}
    \label{fig:1}
\end{figure}

\subsection{Dependent Task Scheduling and Multi-Job FL}
Dependency-constrained edge workloads couple offloading, scheduling, latency, energy, and reliability~\cite{liu2023dependent}. Because precedence constraints make the resulting optimization difficult~\cite{zhao2021ServiceCaching}, learning-based controllers such as federated deep Q-networks, Lyapunov methods, and bandit methods have been explored~\cite{10713971,dai2023offloading}. Multi-job FL scheduling further shows that participant assignment should be treated as a system-level scheduling problem rather than a purely statistical choice~\cite{zhou2022efficient,9685698}. However, many multi-job FL studies assume independent jobs or parallel training without explicit service-flow dependencies. A-CoDa differs by considering a DAG of dependent FL tasks and measuring completion time through layer-wise critical execution, so delays in upstream tasks can propagate to downstream tasks.

\subsection{Client Participation and Cluster-Based FL}
Client selection has been studied from several complementary perspectives. Deadline-aware selection, tier-based grouping, and utility--speed scoring reduce stragglers and improve time-to-accuracy~\cite{8761315,chai2020tifl,lai2021oort}. Recent dynamic selection methods also use loss, model updates, or distribution metrics to prioritize statistically useful clients~\cite{li2024adafl,jimenez2025psi,ren2025dynamic}. Clustered FL manages statistical and resource heterogeneity by grouping clients with similar data or performance profiles~\cite{wang2025fedccs,10074237,10197242}. These methods demonstrate the value of participation control, but they typically optimize a single FL job, an independent multi-job setting, or a static clustered structure.

A-CoDa treats the cluster structure as a stable low-complexity backbone and uses FedMIX for round-level flexibility. FedMIX extends loss--latency selection with an uncertainty bonus and an adaptive probing budget, making external exploration a controlled management resource rather than a fixed overhead. This is particularly relevant for dependent FL flows, where over-exploration may increase tail latency while under-exploration may increase sampling bias.

\subsection{Summary}
Prior studies have advanced non-IID modeling, edge scheduling, multi-job FL, and clustered participation control separately. A-CoDa integrates these directions by combining LDD-driven balanced clustering, FedMIX-based adaptive participation management, and DAG-aware execution. Its analysis links LDD-induced sampling mismatch, residual distribution shift, local-SGD drift, stochastic variance, and adaptive probing budgets to sufficient round and completion-time bounds, while the algorithmic design connects this guidance to practical loss--latency--uncertainty selection.

\section{System Model}
\label{sec:system_model}

\subsection{Multi-Task Federated Learning over IoT-Edge Systems}
\label{subsec:mtfl_system}

We consider a multi-task federated learning (FL) system deployed over an IoT-edge infrastructure, consisting of one edge server and $U$ client devices indexed by
$\mathcal U=\{1,\ldots,U\}$. Client $u\in\mathcal U$ holds a local dataset $D_u$ with size $|D_u|$. A set of learning tasks is denoted by
$\mathcal V=\{1,\ldots,V\}$. For task $v\in\mathcal V$, let $D_{v,u}\subseteq D_u$ be the subset of client $u$'s data associated with task $v$. We define
\begin{equation}
    n_{v,u}\triangleq |D_{v,u}|,\qquad
    n_v\triangleq \sum_{u\in\mathcal U} n_{v,u},
\end{equation}
and the task-eligible client pool as
\begin{equation}
    \mathcal U_v=\{u\in\mathcal U:n_{v,u}>0\}.
\end{equation}
The normalized data weight of client $u$ for task $v$ is
\begin{equation}
    p_{v,u}=\frac{n_{v,u}}{n_v},\qquad u\in\mathcal U_v.
\end{equation}

The tasks are organized as a directed acyclic graph (DAG)
\begin{equation}
    \mathcal G=(\mathcal V,\mathcal E),
\end{equation}
where each node $v\in\mathcal V$ represents a learning task, and each directed edge $(v,q)\in\mathcal E$ indicates that task $v$ must be completed before task $q$ starts. To avoid notation conflict with the number of local SGD steps, we use $\mathcal E$ to denote the edge set of the DAG. The DAG is partitioned into $L$ execution layers, denoted by
\begin{equation}
    \mathcal S=\{\mathcal S_1,\ldots,\mathcal S_L\},
\end{equation}
where tasks in the same layer can be executed in parallel, while different layers are executed sequentially. The execution proceeds from layer $1$ to layer $L$.

All clients are partitioned into $N$ disjoint clusters
\begin{equation}
    \mathcal C=\{\mathcal C_1,\ldots,\mathcal C_N\},
\end{equation}
where $\mathcal C_i\subseteq\mathcal U$, $\mathcal C_i\cap\mathcal C_j=\emptyset$ for $i\neq j$, and $\bigcup_{i=1}^{N}\mathcal C_i=\mathcal U$. Let $i(u)$ denote the cluster index of client $u$, i.e., $u\in\mathcal C_{i(u)}$. For each task $v$, the cluster assigned as its primary serving cluster is denoted by
\begin{equation}
    \phi(v)\in\{1,\ldots,N\}.
\end{equation}
For tasks executed in the same DAG layer, the cluster-task assignment satisfies
\begin{equation}
    \phi(v)\neq \phi(q),\qquad \forall v\neq q,\; v,q\in\mathcal S_l,\; l=1,\ldots,L,
    \label{eq:cluster_task_injective}
\end{equation}
which means that one cluster serves at most one task at the same time.

For task $v$, the in-cluster eligible client pool is defined as
\begin{equation}
    \mathcal M_v
    =\mathcal U_v\cap \mathcal C_{\phi(v)},
    \qquad
    M_v=|\mathcal M_v|.
    \label{eq:in_cluster_pool}
\end{equation}
The out-of-cluster eligible client pool is defined as
\begin{equation}
    \mathcal O_v
    =\mathcal U_v\setminus \mathcal C_{\phi(v)}.
    \label{eq:out_cluster_pool}
\end{equation}

For task $v$, the global loss function is defined as the weighted sum of local losses over its task-eligible client pool:
\begin{equation}
    F_v(w)=\sum_{u\in\mathcal U_v}p_{v,u}F_{v,u}(w),
    \label{eq:global_obj}
\end{equation}
where $F_{v,u}(w)$ is the empirical loss of model parameter $w$ on client $u$ for task $v$. The corresponding global minimizer and local minimizer are denoted by
\begin{equation}
    w_v^*=\arg\min_w F_v(w),
    \qquad
    w_{v,u}^*=\arg\min_w F_{v,u}(w).
\end{equation}

Each task $v$ has an initial global model $w_v^{(0)}$ and a task-specific target requirement $\tau_v$. Let $R_v$ denote the number of communication rounds required for task $v$ to meet its target requirement. At communication round $r=1,\ldots,R_v$, the edge server broadcasts $w_v^{(r-1)}$ to the selected active client set $\mathcal A_v^{(r)}$. Each active client $u\in\mathcal A_v^{(r)}$ initializes
\begin{equation}
    w_{v,u}^{(r,0)}=w_v^{(r-1)}
\end{equation}
and performs $E_{\mathrm{loc}}$ local SGD steps:
\begin{equation}
    w_{v,u}^{(r,e+1)}
    =
    w_{v,u}^{(r,e)}
    -
    \eta_v^{(r)}
    g_{v,u}^{(r,e)},
    \qquad e=0,\ldots,E_{\mathrm{loc}}-1,
    \label{eq:local_sgd}
\end{equation}
where $\eta_v^{(r)}$ is the learning rate and $g_{v,u}^{(r,e)}$ is a stochastic gradient of $F_{v,u}$.

After local training, the server aggregates the uploaded local models using data-size-aware weights:
\begin{equation}
    w_v^{(r)}
    =
    \sum_{u\in\mathcal A_v^{(r)}}
    \alpha_{v,u}^{(r)}
    w_{v,u}^{(r,E_{\mathrm{loc}})},
    \label{eq:fedavg_update}
\end{equation}
where
\begin{equation}
    \alpha_{v,u}^{(r)}
    =
    \frac{n_{v,u}}
    {\sum_{j\in\mathcal A_v^{(r)}}n_{v,j}},
    \qquad
    u\in\mathcal A_v^{(r)}.
    \label{eq:active_agg_weight}
\end{equation}

To use a unified notation for both theoretical analysis and numerical evaluation, we define the task progress metric $Q_v^{(r)}$. In the loss-domain analysis, it is given by the relative loss reduction
\begin{equation}
    Q_v^{(r)}
    =
    1-\frac{F_v(w_v^{(r)})}{F_v(w_v^{(0)})}.
    \label{eq:loss_reduction_metric}
\end{equation}
In the experiments, $Q_v^{(r)}$ is instantiated as the test classification accuracy of task $v$ after round $r$. Task $v$ is regarded as completed when
\begin{equation}
    Q_v^{(R_v)}\ge \tau_v.
    \label{eq:task_completion}
\end{equation}

\subsection{Data Heterogeneity}
\label{subsec:heterogeneity}

In realistic IoT-edge FL systems, client data are generally non-IID. We characterize task-wise data heterogeneity using label-distribution divergence (LDD), defined as the $\ell_1$ distance between empirical label distributions. Let $\mathcal Y=\{1,\ldots,Y\}$ denote the label space. For task $v$, client $u$ induces an empirical label distribution $P_u^{(v)}(y)$ over $\mathcal Y$.

The global label distribution of task $v$ is
\begin{equation}
    P_g^{(v)}(y)
    =
    \sum_{u\in\mathcal U_v}
    p_{v,u}P_u^{(v)}(y),
    \qquad y\in\mathcal Y.
    \label{eq:global_label_dist}
\end{equation}
The client-level LDD of client $u$ for task $v$ is
\begin{equation}
    \Delta_u^{(v)}
    =
    \left\|P_u^{(v)}-P_g^{(v)}\right\|_1
    =
    \sum_{y=1}^{Y}
    \left|
    P_u^{(v)}(y)-P_g^{(v)}(y)
    \right|.
    \label{eq:client_ldd}
\end{equation}

For cluster $\mathcal C_i$, define the number of task-$v$ samples in the cluster as
\begin{equation}
    n_{v,\mathcal C_i}
    =
    \sum_{u\in\mathcal C_i}n_{v,u}.
\end{equation}
When $n_{v,\mathcal C_i}>0$, the aggregated label distribution of cluster $\mathcal C_i$ for task $v$ is
\begin{equation}
    P_i^{(v)}(y)
    =
    \frac{1}{n_{v,\mathcal C_i}}
    \sum_{u\in\mathcal C_i}
    n_{v,u}P_u^{(v)}(y),
    \qquad y\in\mathcal Y.
    \label{eq:cluster_label_dist}
\end{equation}
The cluster-level LDD is then defined as
\begin{equation}
    \Delta_{\mathcal C_i}^{(v)}
    =
    \left\|P_i^{(v)}-P_g^{(v)}\right\|_1
    =
    \sum_{y=1}^{Y}
    \left|
    P_i^{(v)}(y)-P_g^{(v)}(y)
    \right|.
    \label{eq:cluster_ldd}
\end{equation}

For client $u$, the multi-task LDD aggregates task-wise discrepancies according to the client's task data shares:
\begin{equation}
    \Delta_u^{\mathrm{mt}}
    =
    \sum_{v\in\mathcal V}
    \frac{n_{v,u}}{n_u}
    \left\|P_u^{(v)}-P_g^{(v)}\right\|_1,
    \qquad
    n_u=\sum_{v\in\mathcal V}n_{v,u}.
    \label{eq:client_mt_ldd}
\end{equation}
Similarly, the multi-task LDD of cluster $\mathcal C_i$ is
\begin{equation}
    \Delta_{\mathcal C_i}^{\mathrm{mt}}
    =
    \sum_{v\in\mathcal V}
    \frac{n_{v,\mathcal C_i}}{n_{\mathcal C_i}}
    \left\|P_i^{(v)}-P_g^{(v)}\right\|_1,
    \qquad
    n_{\mathcal C_i}=\sum_{v\in\mathcal V}n_{v,\mathcal C_i}.
    \label{eq:cluster_mt_ldd}
\end{equation}
The overall cluster-level multi-task LDD is
\begin{equation}
    \bar{\Delta}_{\mathcal C}^{\mathrm{mt}}
    =
    \sum_{i=1}^{N}
    \frac{n_{\mathcal C_i}}
    {\sum_{j=1}^{N}n_{\mathcal C_j}}
    \Delta_{\mathcal C_i}^{\mathrm{mt}}.
    \label{eq:overall_mt_ldd}
\end{equation}

Since the active client set may include both in-cluster and out-of-cluster clients, we further define the round-level active-set LDD. Let
\begin{equation}
    n_{v,\mathcal C_i}^{(r)}
    =
    \sum_{u\in\mathcal A_v^{(r)}\cap\mathcal C_i}n_{v,u},
    \qquad
    n_v^{(r)}
    =
    \sum_{u\in\mathcal A_v^{(r)}}n_{v,u}.
\end{equation}
The active-set LDD of task $v$ at round $r$ is
\begin{equation}
    \bar{\Delta}_v^{(r)}
    =
    \sum_{i=1}^{N}
    \frac{n_{v,\mathcal C_i}^{(r)}}{n_v^{(r)}}
    \Delta_{\mathcal C_i}^{(v)}.
    \label{eq:round_active_ldd}
\end{equation}
This quantity captures the heterogeneity level induced by the actual client participation decision at round $r$.

\subsection{Transmission and Computation Time}
\label{subsec:latency}

In each FL round, the latency incurred by client $u$ for task $v$ consists of local computation time and uplink communication time. Let $\kappa_v$ denote the computational complexity coefficient of task $v$, and let $f_u$ denote the effective computing rate of client $u$. The local computation time of client $u$ in round $r$ is
\begin{equation}
    t_{u,v}^{(r),\mathrm{comp}}
    =
    \frac{E_{\mathrm{loc}}\,n_{v,u}\,\kappa_v}{f_u},
    \label{eq:comp_time}
\end{equation}
where $E_{\mathrm{loc}}$ is the number of local SGD steps. The effective computing rate is defined as
\begin{equation}
    f_u=\frac{f_u^{\mathrm{clock}}}{C_u},
\end{equation}
where $f_u^{\mathrm{clock}}$ is the CPU clock frequency and $C_u$ is the number of CPU cycles required to process one unit of data.

To avoid conflict with the external probing budget, we use $W$ to denote the wireless channel bandwidth. According to Shannon's capacity formula, the uplink transmission rate of client $u$ in round $r$ is
\begin{equation}
    R_u^{(r)}
    =
    W\log_2
    \left(
    1+\frac{p_u h_u^{(r)}}{\sigma^2}
    \right),
    \label{eq:uplink_rate}
\end{equation}
where $p_u$ is the transmission power of client $u$, $h_u^{(r)}$ is the instantaneous uplink channel gain, and $\sigma^2$ is the receiver noise power.

Let $S_v$ be the model size of task $v$ in bits. The uplink communication time is
\begin{equation}
    t_{u,v}^{(r),\mathrm{comm}}
    =
    \frac{S_v}{R_u^{(r)}}.
    \label{eq:comm_time}
\end{equation}
Therefore, the total latency of client $u$ in round $r$ is
\begin{equation}
    t_{u,v}^{(r)}
    =
    t_{u,v}^{(r),\mathrm{comp}}
    +
    t_{u,v}^{(r),\mathrm{comm}}
    =
    \frac{E_{\mathrm{loc}}\,n_{v,u}\,\kappa_v}{f_u}
    +
    \frac{S_v}{R_u^{(r)}}.
    \label{eq:client_latency}
\end{equation}

Since synchronous FL aggregation waits for all selected clients, the round latency of task $v$ is determined by the slowest active client:
\begin{equation}
    T_v^{(r)}
    =
    \max_{u\in\mathcal A_v^{(r)}}t_{u,v}^{(r)}.
    \label{eq:round_latency}
\end{equation}
The total training time of task $v$ is
\begin{equation}
    T_v
    =
    \sum_{r=1}^{R_v}T_v^{(r)}.
    \label{eq:task_total_time}
\end{equation}

\subsection{Dynamic Clustering and FedMIX-Based Client Selection}
\label{subsec:dynamic_selection}

The proposed scheduling framework uses dynamic client participation to balance learning progress, system latency, exploration, and cross-cluster diversity. For each task $v$, the active client set at round $r$ consists of two parts:
\begin{equation}
    \mathcal A_v^{(r)}
    =
    \mathcal I_v^{(r)}
    \cup
    \mathcal R_v^{(r)},
    \label{eq:active_set}
\end{equation}
where $\mathcal I_v^{(r)}$ is the intra-cluster selected set from the primary cluster $\mathcal C_{\phi(v)}$, and $\mathcal R_v^{(r)}$ is the inter-cluster recruited set from other clusters.

\subsubsection{Intra-Cluster Candidate Selection}

For task $v$, the intra-cluster candidate pool is $\mathcal M_v$ defined in~\eqref{eq:in_cluster_pool}. Given the intra-cluster participation ratio $\rho\in(0,1]$ and the per-round participation capacity $K_{\mathrm{cap}}$, the number of intra-cluster participants is
\begin{equation}
    K_{v,\mathrm{in}}
    =
    \min\left\{
    \left\lceil \rho M_v \right\rceil,
    K_{\mathrm{cap}}
    \right\}.
    \label{eq:intra_num}
\end{equation}
Thus,
\begin{equation}
    \mathcal I_v^{(r)}\subseteq \mathcal M_v,
    \qquad
    |\mathcal I_v^{(r)}|=K_{v,\mathrm{in}}.
    \label{eq:intra_constraint}
\end{equation}

At the beginning of round $r$, the server maintains historical estimates for each candidate client, including an exponential moving average (EMA) of its observed loss and latency. Let $\hat{\ell}_{v,u}^{(r)}$ and $\hat{t}_{v,u}^{(r)}$ denote the EMA loss and EMA latency of client $u$ for task $v$ before round $r$. After client $u$ participates in round $r$, the server observes its local loss $\ell_{v,u}^{(r)}$ and latency $t_{u,v}^{(r)}$, and updates
\begin{equation}
    \hat{x}_{v,u}^{(r+1)}
    =
    \begin{cases}
    (1-\beta_x)\hat{x}_{v,u}^{(r)}+\beta_x x_{v,u}^{(r)},
    & u\in\mathcal A_v^{(r)},\\
    \hat{x}_{v,u}^{(r)},
    & u\notin\mathcal A_v^{(r)},
    \end{cases}
    \qquad
    x\in\{\ell,t\},
    \label{eq:ema_update}
\end{equation}
where $\beta_x\in(0,1]$ is the EMA smoothing factor.

Let $m_{v,u}^{(r)}$ be the number of times client $u$ has participated in task $v$ before round $r$:
\begin{equation}
    m_{v,u}^{(r)}
    =
    \sum_{s=1}^{r-1}
    \mathbf{1}\{u\in\mathcal A_v^{(s)}\}.
    \label{eq:participation_count}
\end{equation}
The uncertainty bonus used for exploration is defined as
\begin{equation}
    b_{v,u}^{(r)}
    =
    \sqrt{
    \frac{\log(r+1)}
    {m_{v,u}^{(r)}+1}
    }.
    \label{eq:ucb_bonus}
\end{equation}

Let $\tilde{\ell}_{v,u}^{(r)}$, $\tilde{t}_{v,u}^{(r)}$, and $\tilde{b}_{v,u}^{(r)}$ denote the normalized EMA loss, normalized EMA latency, and normalized uncertainty bonus, respectively. All three are normalized to $[0,1]$ within the corresponding candidate pool. The FedMIX utility score of client $u$ for task $v$ at round $r$ is
\begin{equation}
    U_{v,u}^{(r)}
    =
    \tilde{\ell}_{v,u}^{(r)}
    -
    \lambda_T\tilde{t}_{v,u}^{(r)}
    +
    \lambda_U\tilde{b}_{v,u}^{(r)},
    \label{eq:fedmix_utility}
\end{equation}
where $\lambda_T\ge0$ controls the latency penalty and $\lambda_U\ge0$ controls the exploration strength. A larger $U_{v,u}^{(r)}$ indicates that client $u$ is more preferred because it has higher potential learning contribution, lower latency cost, or larger uncertainty.

The intra-cluster selected set $\mathcal I_v^{(r)}$ is chosen as the top-$K_{v,\mathrm{in}}$ clients in $\mathcal M_v$ according to $U_{v,u}^{(r)}$:
\begin{equation}
    \mathcal I_v^{(r)}
    =
    \operatorname{TopK}_{K_{v,\mathrm{in}}}
    \left(
    \left\{
    U_{v,u}^{(r)}:u\in\mathcal M_v
    \right\}
    \right).
    \label{eq:intra_topk}
\end{equation}

\subsubsection{Adaptive Inter-Cluster Recruitment}

To mitigate residual imbalance within the primary cluster and to improve exploration across clusters, the server may recruit a limited number of task-eligible clients from other clusters. The external candidate pool is $\mathcal O_v$ defined in~\eqref{eq:out_cluster_pool}. The number of external clients is controlled by an adaptive probing budget $B_v^{(r)}$.

Let the remaining task gap before round $r$ be
\begin{equation}
    g_v^{(r)}
    =
    \left[\tau_v-Q_v^{(r-1)}\right]_+,
    \label{eq:task_gap}
\end{equation}
where $[x]_+=\max\{x,0\}$. Let $H_p$ be the progress-checking window. The recent progress of task $v$ is
\begin{equation}
    d_v^{(r)}
    =
    Q_v^{(r-1)}
    -
    Q_v^{(\max\{0,r-H_p\})}.
    \label{eq:recent_progress}
\end{equation}
Given a stagnation threshold $\delta_v\ge0$, define $I_{\mathrm{stag}}^{(r)}=\mathbf{1}\{d_v^{(r)}<\delta_v\}$. Let $I_{\mathrm{unc}}^{(r)}$ indicate high external-cluster uncertainty and let $I_{\mathrm{cong}}^{(r)}$ indicate latency congestion. The probing demand is summarized as
\begin{equation}
    \xi_v^{(r)}
    =
    I_{\mathrm{stag}}^{(r)}
    \min\left\{1,\frac{g_v^{(r)}}{\tau_v}\right\},
    \label{eq:probing_demand}
\end{equation}
which is used together with uncertainty and congestion signals to adapt the external probing budget:
\begin{equation}
    B_v^{(r)}
    =
    \operatorname{clip}\left(
    B_0+I_{\mathrm{stag}}^{(r)}+I_{\mathrm{unc}}^{(r)}-I_{\mathrm{cong}}^{(r)},
    B_{\min},B_{\max}
    \right),
    \label{eq:adaptive_budget}
\end{equation}
where $B_{\min}$ and $B_{\max}$ are the minimum and maximum external probing budgets. This notation distinguishes the probing budget $B_v^{(r)}$ from the wireless bandwidth $W$ in~\eqref{eq:uplink_rate}.

The recruited set satisfies
\begin{equation}
    \mathcal R_v^{(r)}
    \subseteq
    \mathcal O_v,
    \qquad
    |\mathcal R_v^{(r)}|
    \le
    B_v^{(r)}.
    \label{eq:external_constraint}
\end{equation}
External clients are selected according to the same FedMIX utility score in~\eqref{eq:fedmix_utility}, subject to the external budget.

To prevent a small group of external clients from being repeatedly associated with the same task, we define the consecutive external-recruitment counter
\begin{equation}
    h_{v,u}^{(r)}
    =
    \begin{cases}
    h_{v,u}^{(r-1)}+1,
    & u\in\mathcal R_v^{(r)},\\
    0,
    & u\notin\mathcal R_v^{(r)}.
    \end{cases}
    \label{eq:external_counter}
\end{equation}
The maximum consecutive recruitment length is constrained by
\begin{equation}
    h_{v,u}^{(r)}\le H,
    \qquad
    \forall u\in\mathcal O_v,
    \label{eq:external_hold_constraint}
\end{equation}
where $H$ is a positive integer.

Combining intra-cluster selection and inter-cluster recruitment, the overall active set satisfies
\begin{equation}
    \mathcal A_v^{(r)}
    =
    \mathcal I_v^{(r)}
    \cup
    \mathcal R_v^{(r)},
    \qquad
    |\mathcal A_v^{(r)}|\le K_{\mathrm{cap}}.
    \label{eq:active_capacity}
\end{equation}

\subsection{Optimization Problem Formulation}
\label{subsec:objective}

The overall objective is to minimize the total wall-clock training time required to complete all tasks in the DAG. Since tasks within the same layer can be executed in parallel, while layers are executed sequentially, the total completion time is the sum of the maximum task completion time in each layer.

The joint scheduling, clustering, and client selection problem can be formulated as
\begin{equation}
\begin{aligned}
    \min_{\phi,\{\mathcal I_v^{(r)},\mathcal R_v^{(r)}\}}
    \quad
    T_{\mathrm{total}}
    &=
    \sum_{l=1}^{L}
    \max_{v\in\mathcal S_l}
    T_v&=
    \sum_{l=1}^{L}
    \max_{v\in\mathcal S_l}
    \sum_{r=1}^{R_v}
    \max_{u\in\mathcal A_v^{(r)}}
    t_{u,v}^{(r)}.
\end{aligned}
\label{eq:total_objective}
\end{equation}
The problem is subject to the following constraints for all $v\in\mathcal V$ and $r=1,\ldots,R_v$:
\begin{subequations}
\label{eq:optimization_constraints}
\begin{align}
    & Q_v^{(R_v)}\ge \tau_v,
    \label{eq:constraint_completion}
    \\
    & \phi(v)\in\{1,\ldots,N\},
    \label{eq:constraint_cluster_assignment}
    \\
    & \phi(v)\neq\phi(q),
    \quad
    \forall v\neq q,\; v,q\in\mathcal S_l,\; l=1,\ldots,L,
    \label{eq:constraint_layer_assignment}
    \\
    & \mathcal I_v^{(r)}\subseteq\mathcal M_v,
    \qquad
    |\mathcal I_v^{(r)}|=K_{v,\mathrm{in}},
    \label{eq:constraint_intra_selection}
    \\
    & \mathcal R_v^{(r)}\subseteq\mathcal O_v,
    \qquad
    |\mathcal R_v^{(r)}|\le B_v^{(r)},
    \label{eq:constraint_external_selection}
    \\
    & \mathcal A_v^{(r)}
    =
    \mathcal I_v^{(r)}
    \cup
    \mathcal R_v^{(r)},
    \label{eq:constraint_active_set}
    \\
    & |\mathcal A_v^{(r)}|\le K_{\mathrm{cap}},
    \label{eq:constraint_capacity}
    \\
    & h_{v,u}^{(r)}\le H,
    \qquad
    \forall u\in\mathcal O_v.
    \label{eq:constraint_external_hold}
\end{align}
\end{subequations}

The above formulation is combinatorial because it jointly involves DAG-aware task scheduling, cluster-task assignment, intra-cluster client selection, and adaptive inter-cluster recruitment. Therefore, the proposed algorithms solve it through a layered procedure: first assigning tasks and clusters according to the DAG structure, then performing FedMIX-based intra-cluster selection, and finally activating adaptive inter-cluster probing when task progress stagnates or the remaining task gap is large.

\section{Convergence Behavior Under Sampling Bias, Variance, and Latency}
\label{sec:convergence-ldd-loss-inter}

This section analyzes how dynamic participation affects the training time of a dependent FL flow. The purpose of the analysis is not to predict the exact accuracy-threshold rounds observed in Sec.~VI. Instead, it provides a sufficient loss-domain design bound that clarifies how sampling mismatch, LDD, residual distribution shift, local-SGD drift, stochastic variance, and adaptive probing budgets enter the round and completion-time behavior. We first analyze one arbitrary task $v$ and then compose the task-wise bounds across the DAG layers.

For notational simplicity, we omit the task index when no confusion arises and write $F\equiv F_v$, $F_u\equiv F_{v,u}$, $\pi_u\equiv p_{v,u}$, $P_u\equiv P_u^{(v)}$, $P\equiv P_g^{(v)}$, and $\Delta_u\equiv\|P_u-P\|_1$. The target weights $\{\pi_u\}_{u\in\mathcal U_v}$ form the data-weighted distribution of task $v$. At round $r$, the server selects an active set $S_r=\mathcal A_v^{(r)}$. To account for data-size-aware aggregation, we define the induced aggregation distribution
\begin{equation}
q_r(u)=\mathbb E\left[\alpha_{v,u}^{(r)}\mathbf 1\{u\in S_r\}\mid w^{(r)}\right],
\label{eq:qr-agg-law}
\end{equation}
where $\alpha_{v,u}^{(r)}$ is given in~\eqref{eq:active_agg_weight}. Thus $q_r(u)\ge0$ and $\sum_{u\in\mathcal U_v}q_r(u)=1$. The selection law $q_r$ is allowed to be time-varying and history-dependent because FedMIX uses EMA loss, EMA latency, uncertainty bonuses, and the adaptive probing budget in~\eqref{eq:fedmix_utility}--\eqref{eq:adaptive_budget}. Let $\mathcal Q$ denote the admissible class of such aggregation distributions satisfying the participation constraints.

\subsection{Assumptions}

\begin{definition}[$L$-smoothness and $\mu$-PL condition]
\label{ass1}
Each local objective $F_u$ is $L$-smooth on an admissible parameter domain $\mathcal W$, and the global objective $F$ satisfies the Polyak--\L{}ojasiewicz (PL) inequality with $\mu>0$:
\begin{equation}
\frac{1}{2}\|\nabla F(w)\|^2\ge \mu\big(F(w)-F(w^*)\big),\qquad w\in\mathcal W .
\label{eq:PL}
\end{equation}
\end{definition}

\begin{definition}[Distribution-gradient regularity with residual shift]
\label{ass2}
For every client $u$ and parameter $w\in\mathcal W$, the client-gradient deviation admits
\begin{equation}
\|\nabla F_u(w)-\nabla F(w)\|
\le C_g\|P_u-P\|_1+\chi_u(w),
\label{eq:dist-lip}
\end{equation}
where $C_g\ge0$ captures the dominant label-marginal effect and $\chi_u(w)\ge0$ is a residual term that accounts for covariate or feature-conditional shift not explained by label histograms. We assume
$\bar\chi\triangleq\sup_{q\in\mathcal Q,w\in\mathcal W}\sum_u |q(u)-\pi_u|\chi_u(w)<\infty$.
\end{definition}

Assumption~\ref{ass2} is deliberately stated as a modeling regularity condition rather than as a claim that LDD fully determines gradient mismatch. Under a pure label-mixture model, the residual term can be small; under covariate shift, $\bar\chi$ records the unmodeled component. This addresses the fact that LDD is a lightweight proxy used by the management layer, not a complete description of all data heterogeneity.

\begin{definition}[Loss-driven stochastic variance]
\label{ass3}
For each client $u$ and each $w\in\mathcal W$, the stochastic gradient noise satisfies
\begin{equation}
\mathbb E\|\nabla f(w;\xi)-\nabla F_u(w)\|^2
\le \alpha\ell_u(w)+\beta,
\label{eq:growth}
\end{equation}
where $\alpha,\beta\ge0$. Moreover, $0\le \ell_u(w)\le \ell_{\max}<\infty$ on $\mathcal W$.
\end{definition}

\begin{definition}[Local-SGD drift and aggregate variance]
\label{ass4}
Selected clients perform $E_{\mathrm{loc}}$ local SGD steps. Let $g_r$ be the aggregate update direction used by the server, and define
\begin{equation}
d^{(r)}\triangleq \mathbb E[g_r\mid w^{(r)}]-\left(\nabla F(w^{(r)})+b^{(r)}\right),
\label{eq:local-drift-def}
\end{equation}
where $b^{(r)}$ is the sampling-bias term defined in~\eqref{eq:bias-def}. We assume $\|\nabla F(w)\|^2\le G^2$ on $\mathcal W$ and
\begin{equation}
\|d^{(r)}\|^2\le \bar D_E,
\qquad
\bar D_E=0\;\text{when}\;E_{\mathrm{loc}}=1 .
\label{eq:local-drift-bound}
\end{equation}
The stochastic component around the conditional mean satisfies
\begin{equation}
\mathbb E\|g_r-\mathbb E[g_r\mid w^{(r)}]\|^2
\le M\|\nabla F(w^{(r)})+b^{(r)}+d^{(r)}\|^2+\bar\Sigma,
\label{eq:var-decomp}
\end{equation}
where $M\ge0$ and $\bar\Sigma\le \alpha\ell_{\max}+\beta$ is a uniform additive variance bound.
\end{definition}

\subsection{Sampling Bias and FedMIX-Induced Mismatch}

The conditional mean of the aggregate update can be decomposed as
\begin{equation}
\mathbb E[g_r\mid w^{(r)}]
=
\nabla F(w^{(r)})+b^{(r)}+d^{(r)},
\end{equation}
where
\begin{equation}
b^{(r)}=
\sum_{u\in\mathcal U_v}\big(q_r(u)-\pi_u\big)\nabla F_u(w^{(r)}).
\label{eq:bias-def}
\end{equation}
Since $\sum_u(q_r(u)-\pi_u)=0$, Assumption~\ref{ass2} gives
\begin{align}
\|b^{(r)}\|
&\le
C_g\sum_u |q_r(u)-\pi_u|\Delta_u
+
\sum_u |q_r(u)-\pi_u|\chi_u(w^{(r)}) 
\notag\\
&\le C_g\bar B+\bar\chi,
\label{eq:bt-square}
\end{align}
where
\begin{equation}
\bar B\triangleq
\sup_{q\in\mathcal Q}\sum_u |q(u)-\pi_u|\Delta_u .
\end{equation}

For FedMIX, $\bar B$ can be related to the intra-cluster ratio and the adaptive probing budget. Let $K_{v,\mathrm{in}}=\lceil \rho M_v\rceil$ and let $\delta_v(\rho)$ denote the worst-case LDD-weighted mismatch of the intra-only selection law. Because FedMIX recruits at most $B_v^{(r)}$ new external clients in round $r$, the external aggregation mass is upper bounded by $B_v^{(r)}/(K_{v,\mathrm{in}}+B_v^{(r)})$. Hence the realized FedMIX mismatch obeys
\begin{equation}
\bar B_{\mathrm{FM}}^{(r)}(\rho,B_v^{(r)})
=
\delta_v(\rho)
+
\frac{2B_v^{(r)}}{K_{v,\mathrm{in}}+B_v^{(r)}}\Delta_{\max,v},
\label{eq:bbar-fedmix}
\end{equation}
where $\Delta_{\max,v}=\max_{u\in\mathcal U_v}\Delta_u$. Since $B_v^{(r)}\le B_{\max}$, a trajectory-independent bound is obtained by replacing $B_v^{(r)}$ with $B_{\max}$. The uncertainty term in FedMIX changes the realized law $q_r$ by encouraging under-sampled clients and clusters; the bound above remains valid because it only requires the resulting $q_r$ to satisfy the participation and probing constraints.

Combining~\eqref{eq:bt-square} with Assumption~\ref{ass4}, the aggregate direction satisfies
\begin{equation}
\|\nabla F(w^{(r)})+b^{(r)}+d^{(r)}\|^2
\le
\bar G_{\mathrm{agg}}^2
\triangleq
3\left(G^2+(C_g\bar B+\bar\chi)^2+\bar D_E\right).
\label{eq:gagg-bound}
\end{equation}
Together with~\eqref{eq:var-decomp}, this yields
\begin{equation}
\sigma_r^2
\triangleq
\mathbb E\|g_r-\mathbb E[g_r\mid w^{(r)}]\|^2
\le
M\bar G_{\mathrm{agg}}^2+\bar\Sigma .
\label{eq:signma}
\end{equation}

\subsection{One-Step Progress and Finite-Time Bound}

By $L$-smoothness of $F$, the server update $w^{(r+1)}=w^{(r)}-\eta g_r$ satisfies
\begin{align}
\mathbb E[F(w^{(r+1)})\mid w^{(r)}]
&\le F(w^{(r)})
-\eta\left\langle \nabla F(w^{(r)}),\mathbb E[g_r\mid w^{(r)}]\right\rangle
\notag\\
&\quad +\frac{L\eta^2}{2}\mathbb E\|g_r\|^2 .
\label{eq:smooth-step}
\end{align}
Using
\begin{equation}
\mathbb E\|g_r\|^2
\le
\|\nabla F(w^{(r)})+b^{(r)}+d^{(r)}\|^2+
\sigma_r^2,
\label{eq:app-second-moment}
\end{equation}
and applying Young's inequality to the cross term yields, for $0<\eta\le 1/(4L)$,
\begin{align}
\mathbb E[F(w^{(r+1)})-F(w^*)]
&\le
\left(1-\frac{\eta\mu}{2}\right)
\mathbb E[F(w^{(r)})-F(w^*)]
\notag\\
&\quad+4\eta\left((C_g\bar B+\bar\chi)^2+\bar D_E\right)
\notag\\
&\quad+\frac{L\eta^2}{2}\left(M\bar G_{\mathrm{agg}}^2+\bar\Sigma\right).
\label{eq:one-step-final}
\end{align}
The complete derivation is given in Appendix~\ref{app:proof-prop-upper-steady}.

\begin{proposition}[Sufficient loss-domain bound]
\label{prop:upper-steady}
Under Assumptions~\ref{ass1}--\ref{ass4}, suppose $0<\eta\le 1/(4L)$. For any admissible selection sequence generated by FedMIX or by any other law in $\mathcal Q$, the iterates satisfy, for all $R\ge1$,
\begin{equation}
\mathbb E[F(w^{(R)})-F(w^*)]
\le
\left(1-\frac{\eta\mu}{2}\right)^R
\left(F(w^{(0)})-F(w^*)\right)+E_\infty,
\label{eq:finite-time}
\end{equation}
where
\begin{align}
E_\infty&=\frac{2}{\eta\mu}\Psi,
\notag\\
\Psi&=
4\eta\left((C_g\bar B+\bar\chi)^2+\bar D_E\right)
+\frac{L\eta^2}{2}\left(M\bar G_{\mathrm{agg}}^2+\bar\Sigma\right).
\label{eq:error-floor-components}
\end{align}
Consequently,
\begin{equation}
\limsup_{R\to\infty}\mathbb E[F(w^{(R)})-F(w^*)]
\le E_\infty .
\label{eq:steady-state}
\end{equation}
\end{proposition}

The result is a sufficient upper bound. It does not state that FedMIX attains the smallest possible number of rounds, nor that the analytical loss-reduction threshold is identical to the empirical classification-accuracy threshold used in Sec.~VI. Its role is to expose the design trade-off: smaller LDD-induced mismatch $\bar B$, smaller residual shift $\bar\chi$, smaller local drift $\bar D_E$, and lower stochastic variance reduce the attainable floor, while larger selected cohorts or external probing can increase per-round latency.

To achieve a target precision $\varepsilon>E_\infty$, it is sufficient that
\begin{equation}
R
\ge
\frac{2}{\eta\mu}
\log\frac{F(w^{(0)})-F(w^*)}{\varepsilon-E_\infty}.
\label{eq:rounds}
\end{equation}

\begin{corollary}[DAG-level completion bound]
\label{cor:dag-completion}
Suppose Assumptions~\ref{ass1}--\ref{ass4} hold for every task $v$, with task-specific constants and error floor $E_{\infty,v}$. For any target $\varepsilon_v>E_{\infty,v}$, if
\begin{equation}
R_v
\ge
\frac{2}{\eta_v\mu_v}
\log\frac{F_v(w_v^{(0)})-F_v(w_v^*)}{\varepsilon_v-E_{\infty,v}},
\label{eq:task-rounds-dag}
\end{equation}
then $\mathbb E[F_v(w_v^{(R_v)})-F_v(w_v^*)]\le\varepsilon_v$. For the layer-sequential DAG in~\eqref{eq:total_objective},
\begin{equation}
T_{\mathrm{DAG}}
\le
\sum_{l=1}^{L}\max_{v\in\mathcal S_l}R_v\bar T_v,
\qquad
\bar T_v=\sup_{1\le r\le R_v}T_v^{(r)} .
\label{eq:dag-time-bound}
\end{equation}
\end{corollary}

\begin{IEEEproof}
Apply Proposition~\ref{prop:upper-steady} to each task and use $T_v^{(r)}\le\bar T_v$. Since tasks in the same layer execute in parallel and layers execute sequentially, summing the maximum duration in each layer gives~\eqref{eq:dag-time-bound}.
\end{IEEEproof}

\subsection{Design Implications for FedMIX}

Equation~\eqref{eq:bbar-fedmix} explains the role of the adaptive probing budget. In a realized round, increasing $B_v^{(r)}$ can improve exploration and reduce stale in-cluster bias, but it also increases the worst-case perturbation term and may raise the tail latency $T_v^{(r)}=\max_{u\in\mathcal A_v^{(r)}}t_{u,v}^{(r)}$. FedMIX therefore uses the budget rule in~\eqref{eq:adaptive_budget} to increase probing when recent progress stagnates or external uncertainty is high, and to reduce probing under latency congestion. The uncertainty bonus in~\eqref{eq:ucb_bonus} does not alter the proof structure; instead, it changes the realized selection law $q_r$ toward under-sampled clients and clusters while remaining inside the admissible class $\mathcal Q$.

A practical choice of the in-cluster ratio $\rho$ should be made jointly with the maximum probing budget. Let
\begin{equation}
K_v^{\max}(\rho)=\lceil \rho M_v\rceil+B_{\max}\le K_{\mathrm{cap}}.
\end{equation}
Under conditionally independent client noise and near-uniform aggregation, a common variance scaling gives
\begin{equation}
\bar\Sigma_{\rho,B_{\max}}
\le
\frac{\alpha\ell_{\max}+\beta}{K_v^{\max}(\rho)}.
\label{eq:rho-variance}
\end{equation}
Thus, larger $\rho$ or larger probing budgets may reduce variance and improve coverage, but also increase synchronous tail latency. This is why the bound should be read as guidance for balancing bias, variance, drift, and latency, rather than as a direct predictor of the empirical curves.

\subsection{Relationship Between LDD and Loss}
\label{sec:upper-bound-steady-state}

The finite-time bound shows that LDD enters the loss-domain analysis through the sampling-bias term $\bar B$. With residual shift included in Assumption~\ref{ass2}, the analysis does not assume that label histograms fully determine client gradients. Instead, LDD is a lightweight management proxy whose explanatory power is conditioned on the residual term $\bar\chi$. If label skew is the dominant source of heterogeneity, reducing intra-cluster LDD can tighten the bound; if covariate shift is strong, the residual term records the limitation of LDD-based coordination. Overall, the theory supports the design of statistically coherent clusters, adaptive exploration, and latency-aware participation, while keeping the claim appropriately framed as a sufficient analytical guideline for dependent FL flow management.

\section{Clustered Client Coordination and Dependency-Aware Scheduling}

Following the previous analysis, our design integrates three components. First, we use LDD-based clustering with greedy balancing to form statistically coherent and size-balanced client groups, reducing selection complexity and providing a stable coordination backbone. Second, we propose FedMIX, an uncertainty-aware intra--inter cluster participation mechanism that refines client participation using loss, latency, and exploration states. FedMIX preserves the efficiency of in-cluster exploitation while using adaptive inter-cluster probing to control cross-cluster diversity and probing overhead. Finally, we adopt a PPO-based DAG scheduler to determine dependency-respecting task execution and reduce workflow-level latency. Together, these components provide a resource-aware and communication-efficient coordination framework for dependent multi-task FL.

\subsection{LDD-based Clustering with Greedy Balancing}

Motivated by the observation that convergence improves when intra-cluster LDD is small and cluster sizes are balanced, we group clients using LDD-based clustering and then apply a lightweight greedy balancing step to equalize cluster cardinalities while preserving statistical homogeneity. Beyond improving statistical coherence, balancing also serves a system-level role: in DAG-structured multi-task FL, highly unbalanced clusters may create load imbalance and unstable per-round latency. The resulting clusters provide a stable structure for the subsequent intra- and inter-cluster coordination. The full procedure is summarized in Algorithm~\ref{alg:cluster_balance}.

The algorithm is computationally efficient after its one-time initialization cost is amortized. Constructing the LDD distance matrix and performing agglomerative clustering require $O(U^2Q+U^2\log U)$ for $U$ clients with $Q$-class label histograms. Although this initialization cost is not eliminated by clustering, the benefit appears in subsequent coordination. A naive non-clustered approach repeatedly searches over the entire client pool, whereas the balanced cluster structure localizes most recurring selection operations within clusters. This reduces the effective post-clustering coordination cost and makes intra-/inter-cluster participation control scalable across rounds and tasks.

\begin{algorithm}[!t]
\caption{LDD-based Clustering with Greedy Balancing}
\label{alg:cluster_balance}
\begin{algorithmic}[1]
\Require Client label histograms $\{h_u\}_{u=1}^U$, desired cluster number $K$
\State Normalize each histogram $h_u$ to a probability vector $p_u = h_u/\sum_d h_u[d]$
\State Compute pairwise LDD distance $D^{\mathrm{LDD}}_{u,u'}=\|p_u-p_{u'}\|_1$
\State Agglomerative clustering on $D^{\mathrm{LDD}}$ $\Rightarrow$ labels $\ell_u\in\{1,\dots,K\}$
\State $\mathcal{C}_i\gets\{u:\ell_u=i\}$,\quad $t_{\text{low}}=\lfloor U/K\rfloor$, $t_{\text{high}}=\lceil U/K\rceil$
\State Assign per-cluster targets $t_i\in\{t_{\text{low}},t_{\text{high}}\}$
\While{clusters not balanced}
  \State Identify oversized clusters $\mathcal{O}=\{i:|\mathcal{C}_i|>t_i\}$ and undersized clusters $\mathcal{I}=\{j:|\mathcal{C}_j|<t_j\}$
  \State Compute cluster medoid $m_i=\arg\min_{u\in\mathcal{C}_i}\sum_{u'\in\mathcal{C}_i}D^{\mathrm{LDD}}_{u,u'}$
  \For{each $i\in\mathcal{O}$ and $u\in\mathcal{C}_i$}
     \For{each $j\in\mathcal{I}$}
        \State Evaluate cost gain $\Delta_{u,i\to j}=D^{\mathrm{LDD}}_{u,m_j}-D^{\mathrm{LDD}}_{u,m_i}$
     \EndFor
  \EndFor
  \State Move the client with the smallest $\Delta_{u,i\to j}$ from $i$ to $j$ and update medoids
\EndWhile
\State \Return balanced clusters $\mathcal{C}=\{\mathcal{C}_1,\dots,\mathcal{C}_K\}$
\end{algorithmic}
\end{algorithm}

\subsection{Multi-task Intra--Inter Cluster Exploration--Exploitation}

FedMIX operates on the balanced LDD-based clusters produced by Algorithm~\ref{alg:cluster_balance}. It preserves the two-stage intra-/inter-cluster structure: the assigned cluster provides the main exploitation backbone, while external clusters are selectively probed to provide controlled diversity. Compared with a pure loss--latency ranking rule, FedMIX makes the exploration--exploitation tradeoff explicit by adding an uncertainty bonus for under-sampled clients and clusters. It also replaces a fixed external probing budget with an adaptive task-round budget.

For each task-client pair $(v,u)$, the server maintains the EMA loss $\hat\ell_{v,u}^{(r)}$, the EMA latency $\hat t_{v,u}^{(r)}$, and the selection count $m_{v,u}^{(r)}$. FedMIX ranks in-cluster clients by the utility in~\eqref{eq:fedmix_utility}, where the loss term captures statistical utility, the latency term penalizes slow participants, and the uncertainty term encourages controlled exploration of less frequently selected clients. For inter-cluster probing, FedMIX applies the same principle at the external-cluster level and uses the adaptive budget $B_v^{(r)}$ in~\eqref{eq:adaptive_budget}. Thus, probing increases when recent task progress stagnates or external uncertainty is high, and decreases under latency congestion.

The overall workflow is illustrated in Fig.~\ref{fig:fedmix-flowchart}. In each round, FedMIX ranks in-cluster clients, adaptively probes external clusters, merges the selected participants under the global cap $K_{\mathrm{cap}}$, and updates EMA statistics after local training. This design keeps the clustered selection backbone while making external exploration a controlled resource rather than a fixed overhead.

\begin{figure}[htbp]
    \centering
    \includegraphics[width=1.0\linewidth]{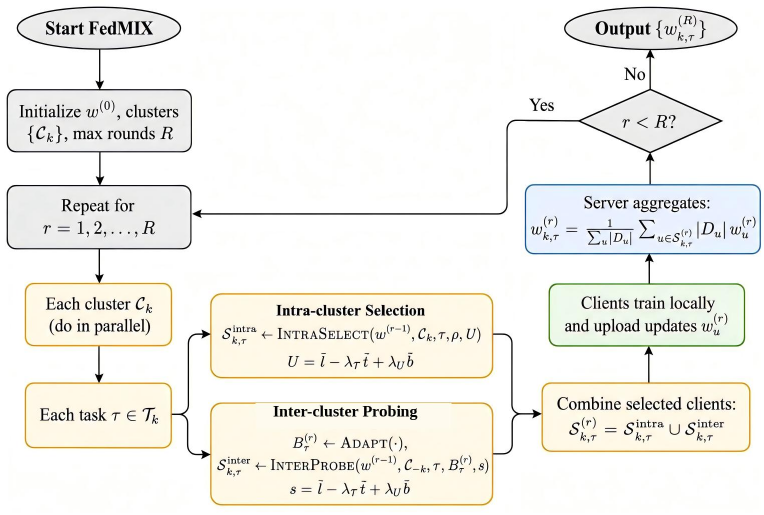}
    \caption{Flowchart of FedMIX. The intra branch ranks in-cluster clients by a loss--latency--uncertainty utility, while the inter branch adaptively probes external clusters using the task-round probing budget $B_v^{(r)}$.}
    \label{fig:fedmix-flowchart}
\end{figure}

\begin{algorithm}[htbp]
\caption{Loss-, Latency-, and Uncertainty-aware Intra-cluster Selector}
\label{alg:fedmix_intra}
\begin{algorithmic}[1]
\Require Cluster $\mathcal{C}_{\phi(v)}$, model $w_v^{(r)}$, ratio $\rho$, eval budget $B_{\mathrm{eval}}$, cap $K_{\mathrm{cap}}$, weights $(\lambda_T,\lambda_U)$, EMA params $(\rho_{\ell},\rho_g,\rho_T)$, externals $\mathcal{R}_v^{(r)}$
\State $\mathcal{M}_v\gets\{u\in\mathcal{C}_{\phi(v)}:|D_{v,u}|>0\}$,\quad $K_{\mathrm{intra}}\gets\lceil \rho|\mathcal{M}_v|\rceil$
\For{$u\in\mathcal{M}_v$}
    \State Estimate local loss $\widehat{\ell}_{u,v}^{(r)}\gets \frac{1}{B_{\mathrm{eval}}}\sum_{b=1}^{B_{\mathrm{eval}}}\mathcal{L}(w_v^{(r)};\mathcal{B}_{u,b})$
    \State Update loss EMA and compute $b_{u,v}^{(r)}=\sqrt{\log(r+1)/(m_{v,u}^{(r)}+1)}$
    \State Compute utility $U_{u,v}^{(r)}=\widetilde{\ell}_{u,v}^{(r)}-\lambda_T\widetilde{t}_{u,v}^{(r)}+\lambda_U\widetilde{b}_{u,v}^{(r)}$
\EndFor
\State $S_{v,\mathrm{intra}}^{(r)}\gets$ top-$K_{\mathrm{intra}}$ clients in $\mathcal{M}_v$ by $U_{u,v}^{(r)}$
\State $\mathcal{A}_v^{(r)}\gets\operatorname{dedup}(S_{v,\mathrm{intra}}^{(r)}\cup\mathcal{R}_v^{(r)})$
\While{$|\mathcal{A}_v^{(r)}|>K_{\mathrm{cap}}$}
    \State Remove one client with the lowest utility; ties are broken by larger latency
\EndWhile
\For{$u\in\mathcal{A}_v^{(r)}$}
    \State Run local SGD, compute improvement, and update $g_{u,v}$, $t_{u,v}^{\mathrm{EMA}}$, and the selection count $m_{v,u}^{(r)}$
\EndFor
\State \Return $S_{v,\mathrm{intra}}^{(r)}$, $\mathcal{A}_v^{(r)}$
\end{algorithmic}
\end{algorithm}

To complement local exploitation, FedMIX adaptively probes external clusters using representative loss, latency, and uncertainty signals. Only a bounded number of external clients are temporarily recruited, and their participation lifetime is capped by $H$ to avoid persistent cross-cluster assignment. Unlike fixed-budget probing, the adaptive budget $B_v^{(r)}$ allows FedMIX to increase external exploration when training stagnates or external clusters are under-explored, while reducing probing under latency congestion. The procedure is summarized in Algorithm~\ref{alg:fedmix_inter}. 

\begin{algorithm}[htbp]
\caption{Adaptive Inter-cluster Probing and Recruitment}
\label{alg:fedmix_inter}
\begin{algorithmic}[1]
\Require Task $v$, assigned cluster $\phi(v)$, model $w_v^{(r)}$, clusters $\{\mathcal{C}_i\}$, active set $\mathcal{A}_v^{(r)}$, lifetimes $\Xi^{(r-1)}$, cap $K_{\mathrm{cap}}$, budget parameters $(B_{\min},B_0,B_{\max})$, horizon $H$, weights $(\lambda_T,\lambda_U)$, threshold $\kappa$, probe budget $B_{\mathrm{probe}}$
\State $\Xi^{(r)}\gets\Xi^{(r-1)}$ and remove expired clients from $\mathcal{A}_v^{(r)}$
\State $\mathcal{J}\gets\{j\neq\phi(v):\exists u\in\mathcal{C}_j,\ |D_{v,u}|>0\}$
\For{$j\in\mathcal{J}$}
    \State Pick representative $u_j=\arg\max_{u\in\mathcal{C}_j}(|D_{v,u}|,-t_{u,v}^{\mathrm{EMA}})$
    \State Estimate probing loss $\widehat{\ell}_{j,v}^{(r)}\gets \frac{1}{B_{\mathrm{probe}}}\sum_b \mathcal{L}(w_v^{(r)};\mathcal{B}_{u_j,b})$
    \State Update cluster loss EMA and compute $b_{j,v}^{(r)}=\sqrt{\log(r+1)/(N_{j,v}^{\mathrm{probe}}+1)}$
    \State Compute external score $s_{j,v}^{(r)}=\widetilde{\ell}_{j,v}^{(r)}-\lambda_T\widetilde{t}_{j,v}^{(r)}+\lambda_U\widetilde{b}_{j,v}^{(r)}$
\EndFor
\If{$\kappa>0$ and $|\mathcal{J}|>1$}
    \State Keep only clusters with $\widehat{\ell}_{j,v}^{(r)}\ge\mu+\kappa\sigma$
    \If{$\mathcal{J}=\emptyset$} \Return $(\mathcal{A}_v^{(r)},\Xi^{(r)})$ \EndIf
\EndIf
\State Compute $B_v^{(r)}=\operatorname{clip}(B_0+I_{\mathrm{stag}}^{(r)}+I_{\mathrm{unc}}^{(r)}-I_{\mathrm{cong}}^{(r)},B_{\min},B_{\max})$
\State $B_{\mathrm{eff}}\gets\min\{B_v^{(r)},|\mathcal{J}|,\max(0,K_{\mathrm{cap}}-|\mathcal{A}_v^{(r)}|)\}$
\State $\mathcal{J}_B\gets$ top-$B_{\mathrm{eff}}$ clusters in $\mathcal{J}$ by $s_{j,v}^{(r)}$
\For{$j\in\mathcal{J}_B$}
    \State Pick $u^*=\arg\max_{u\in\mathcal{C}_j}(|D_{v,u}|,-t_{u,v}^{\mathrm{EMA}})$
    \State $\mathcal{A}_v^{(r)}\gets\mathcal{A}_v^{(r)}\cup\{u^*\}$, $\Xi^{(r)}[u^*]\gets H$, $N_{j,v}^{\mathrm{probe}}\gets N_{j,v}^{\mathrm{probe}}+1$
\EndFor
\State \Return $(\mathcal{A}_v^{(r)},\Xi^{(r)})$
\end{algorithmic}
\end{algorithm}

\subsection{Dependency-Aware Scheduling}
While FedMIX optimizes client participation at the learning layer, the overall training efficiency is fundamentally constrained by the execution order of dependent tasks. To handle this complementary dimension, we adopt a PPO-based DAG scheduling algorithm from our prior conference work~\cite{luo2025cluster}, which minimizes the total latency of dependent multi-task FL execution. The scheduler represents the system state using each task's readiness status, the real-time occupancy of all clusters, and a precomputed processing-time matrix for all task--cluster pairs. A PPO agent then maps this state to dependency-respecting cluster assignments that exploit available parallelism, iteratively refining its policy to reduce the overall makespan. Here, we keep the dependency scheduler fixed across the compared methods and use the A-CoDa/Intra-FL/Inter-FL ablations to isolate the effect of FedMIX-style client coordination. More detailed discussions on the scheduler can be found in our prior work~\cite{luo2025cluster}. A custom callback tracks makespan evolution and records full schedules, ensuring that the learned policy minimizes total latency while satisfying the task completion requirements. 

\section{Numerical Experiment}

\subsection{Experimental Setup} 

We consider four distinctive tasks on classification datasets. The tasks are organized into a three-layer DAG: MNIST digit recognition (Task 1, Layer 1), UCI-HAR smartphone inertial-sensor activity recognition (Task 2, Layer 2)~\cite{anguita2013public}, FashionMNIST product-image classification (Task 3, Layer 2), and Pneumonia X-ray classification from MedMNIST (Task 4, Layer 3)~\cite{yang2023medmnist,kermany2018identifying}. This suite incorporates lightweight vision, wearable sensing, product inspection, and medical-imaging workloads while preserving the same dependency structure. The choice is semantically aligned with IoT-edge deployments: MNIST represents lightweight handwritten-input recognition, UCI-HAR represents mobile and wearable sensing, FashionMNIST represents product-image inspection, and Pneumonia X-ray represents IoMT diagnostic imaging at edge clinics or hospital gateways. All methods use the same PyTorch CNN/GPU validation code path, non-IID data partitions, and wireless setting, so the comparison isolates client selection and dependency-aware scheduling rather than implementation differences. Values of key parameters are summarized in Table~\ref{tab:default-config}.


\begin{table}[t]
\centering
\caption{Default Experimental Configurations}
\label{tab:default-config}
\renewcommand{\arraystretch}{1.15}
\begin{tabularx}{\linewidth}{@{}l X@{}}
\toprule
\textbf{Items} & \textbf{Value} \\
\midrule
\multicolumn{2}{@{}l}{\textbf{Client and System Settings}} \\
CPU frequency, $f_u$ & $f_u \sim \mathcal U(1.2\,\mathrm{GHz},\,2.5\,\mathrm{GHz})$ \\
Computation cost per bit, $C_{D_u}$ & $1000\,\mathrm{cycles/bit}$ \\
Local model size, $S$ & $5\times 10^{7}$ bits \\
System bandwidth, $B$ & $5\,\mathrm{MHz}$ \\
Receiver noise power, $\sigma^2$ & $-107\,\mathrm{dBm}$ \\
Transmit power, $p_u$ & $0.05\,\mathrm{W}$ \\
Accuracy thresholds $\{\tau_1,\ldots,\tau_4\}$ & $\{0.82,\ 0.85,\ 0.72,\ 0.80\}$ \\
\addlinespace[3pt]
\multicolumn{2}{@{}l}{\textbf{Wireless Channel Settings}} \\
Large-scale channel gain, $\bar h_u$ & $\bar h_u \sim \mathrm{exp}\!\big(\lambda = 2.5\times 10^{-7}\big)$ \\
Small-scale fading gain, $h_u$ & $h_u = \bar h_u \exp(z),\quad z \sim \mathcal N(0,0.15^2)$
\\
\bottomrule
\end{tabularx}
\end{table}

\subsection{Baselines}
We compare the total time and convergence performance of our proposed A-CoDa with the following baseline approaches:
\begin{itemize}
  \item \textbf{Cluster-oriented and Dependency-aware Client Selection for FL (CoDa-FL)}: The original framework proposed in our conference paper~\cite{luo2025cluster}, which considers the data heterogeneity across edge clients by performing LDD-based clustering to group clients with similar data distributions.  
  \item \textbf{Intra-cluster Client Selection for FL (Intra-FL)}: This is an ablated version of A-CoDa that dynamically selects clients within each cluster without any cross-cluster interaction.
  \item \textbf{Inter-cluster Client Selection for FL (Inter-FL)}: Inter-FL introduces cross-cluster collaboration and latency-awareness, serving as another ablation variant of A-CoDa that focuses on inter-cluster dynamics only.
  \item \textbf{Population Stability Index for Personalized FL (PSI-PFL)}: The baseline is a static client selection framework that leverages the Population Stability Index to quantify and mitigate client-level data heterogeneity~\cite{jimenez2025psi}.
  \item \textbf{Training-based Dynamic Clustered FL (TDCFL)}: It is a dynamic clustering framework that employs an Adaptive Distribution Similarity Metric combining data features and model updates, with a hierarchical scheduler to improve training efficiency~\cite{ren2025dynamic}.
\end{itemize}

\begin{figure*}[htbp]
  \centering
  \includegraphics[width=0.88\textwidth]{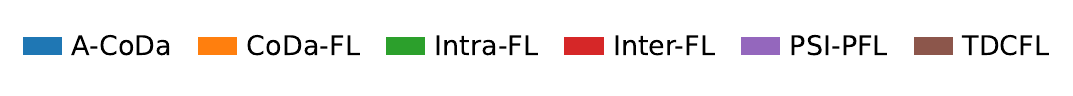}
  \vspace{0.15em}
  \par
  \begin{subfigure}[t]{0.45\textwidth}
    \centering
    \includegraphics[width=\linewidth]{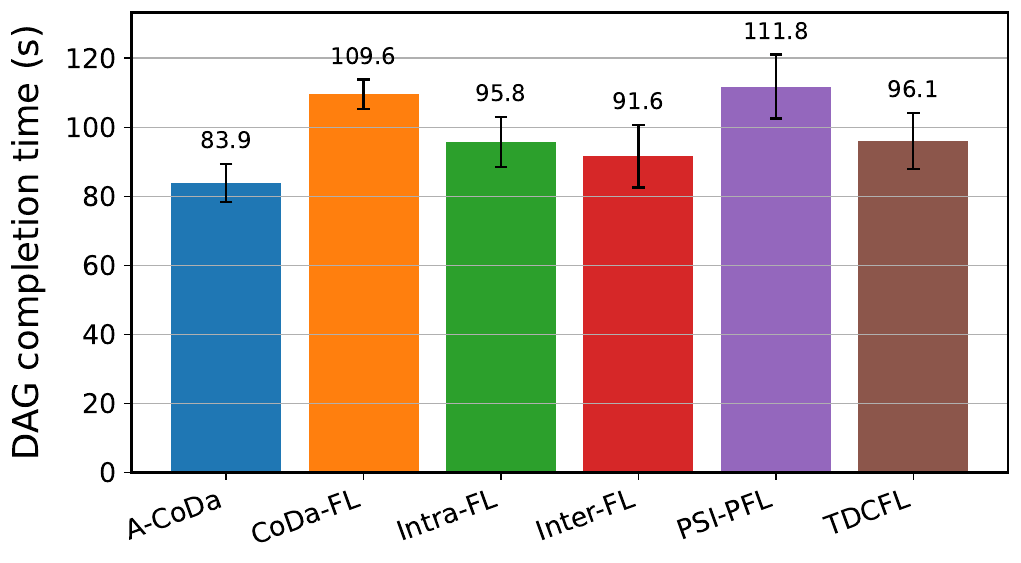}
    \caption{Total time of different approaches.}
    \label{fig:TT1}
  \end{subfigure}
  \begin{subfigure}[t]{0.45\textwidth}
    \centering
    \includegraphics[width=\linewidth]{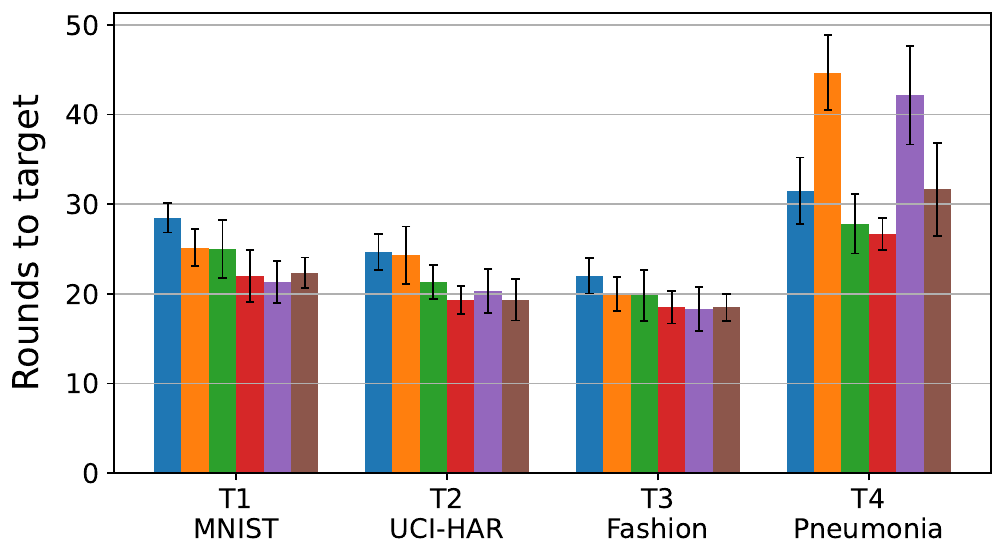}
    \caption{Number of rounds per task of different approaches.}
    \label{fig:round1}
  \end{subfigure}
  \caption{Total time and the number of rounds to convergence. Bars report the mean over six seeds, and error bars denote 95\% confidence intervals.}
  \label{fig:1-100}
\end{figure*}

\begin{figure*}[htbp]
  \centering
  \includegraphics[width=0.88\textwidth]{plots_unified/cnn_gpu/method_legend_bars.pdf}
  \vspace{0.15em}
  \par
  \begin{subfigure}[t]{0.46\textwidth}
    \centering
    \includegraphics[width=\linewidth]{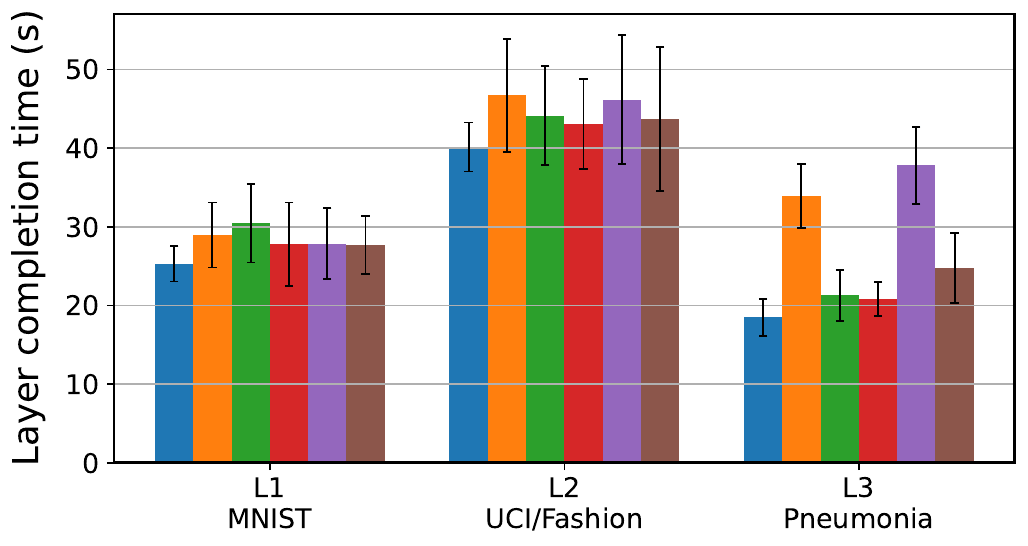}
    \caption{Total time per layer of different approaches.}
    \label{fig:layerTT1}
  \end{subfigure}
  \begin{subfigure}[t]{0.46\textwidth}
    \centering
    \includegraphics[width=\linewidth]{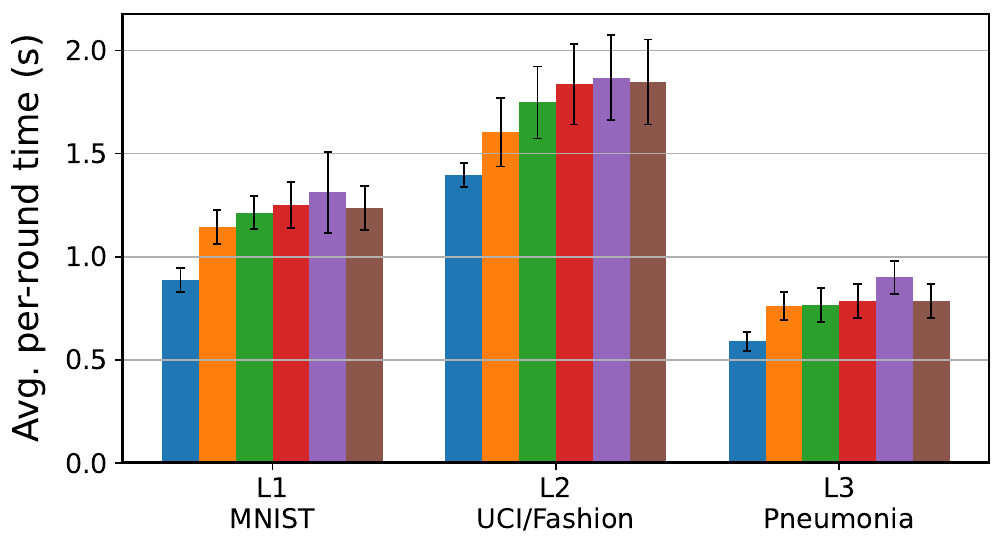}
    \caption{Average time per layer of different methods.}
    \label{fig:layeravgT1}
  \end{subfigure}

  \caption{Layer-wise time performance. Bars report the mean over six seeds, and error bars denote 95\% confidence intervals.}
  \label{fig:2-100}
\end{figure*}

\subsection{Numerical Results}
We evaluate the federated setting with 100 clients, as illustrated in Fig.~\ref{fig:1-100}--\ref{fig:3-100}. The results report target-reaching behavior over six seeds, with 95\% confidence intervals shown in the timing, round, and accuracy figures. The analysis focuses on (1) overall training efficiency, (2) task-level convergence behavior, and (3) layer-wise performance across dependent tasks.

For overall training efficiency, Fig.~\ref{fig:TT1} shows that A-CoDa achieves the lowest end-to-end DAG completion time. Following the DAG execution model in~\eqref{eq:total_objective}, total time is computed as the Layer-1 completion time plus the slower Layer-2 task plus the Layer-3 completion time, rather than as a naive sum of all four task times. Across six seeds, A-CoDa completes the dependent workload in 83.9 s, compared with 91.6 s for the closest baseline Inter-FL, corresponding to an 8.4\% reduction. Relative to CoDa-FL, Intra-FL, PSI-PFL, and TDCFL, A-CoDa reduces completion time by 23.4\%, 12.4\%, 25.0\%, and 12.7\%, respectively.

Fig.~\ref{fig:round1} shows the task-level convergence behavior. All methods reach the four target thresholds. A-CoDa requires 28.5 rounds on MNIST, 24.7 rounds on UCI-HAR, 22.0 rounds on FashionMNIST, and 31.5 rounds on Pneumonia X-ray on average. Although Inter-FL and several baselines require fewer rounds on individual tasks, A-CoDa has a substantially lower target-reaching time because its loss- and latency-aware selection avoids slow selected-client cohorts while retaining useful cross-cluster diversity.

Fig.~\ref{fig:layerTT1} provides the layer-wise timing breakdown. A-CoDa is the fastest method in all three DAG layers under the six-seed mean, with especially large gains in the downstream Pneumonia X-ray layer. Relative to the closest baseline in each layer, A-CoDa reduces Layer-1, Layer-2, and Layer-3 completion time by 8.7\%, 6.7\%, and 11.2\%, respectively. Fig.~\ref{fig:layeravgT1} further shows that the proposed selection policy reduces completion time through the latency of selected cohorts as well as through the number of rounds.

Fig.~\ref{fig:3-100} reports time-aligned accuracy trajectories, where the shaded bands denote 95\% confidence intervals across the six seeds. A-CoDa reaches the required accuracy thresholds on all layers and maintains competitive final accuracy. Some static or dynamic baselines attain slightly higher final accuracy or fewer rounds on individual layers, but they do so with higher end-to-end latency. Overall, the results demonstrate that A-CoDa provides the best completion-time--accuracy tradeoff for the dependent IoT-edge workload.

\begin{figure*}[htbp]
  \centering

  \includegraphics[width=0.88\textwidth]{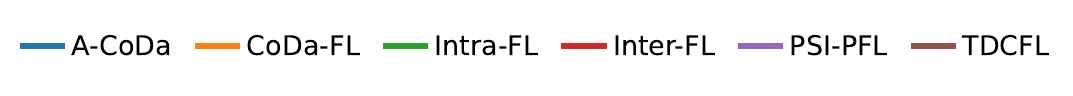}
  \vspace{0.15em}
  \par

  \begin{subfigure}[t]{0.31\textwidth}
    \centering
    \includegraphics[width=\linewidth]{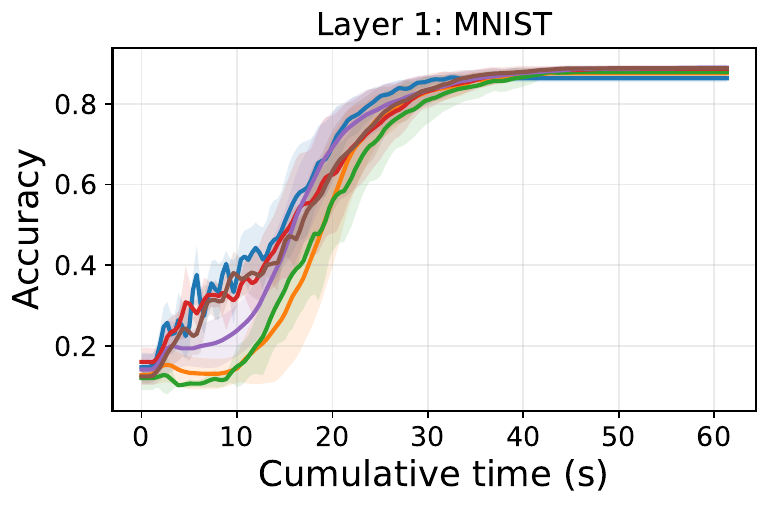}
    \caption{Layer 1.}
    \label{fig:layer1_1}
  \end{subfigure}
  \hfill
  \begin{subfigure}[t]{0.31\textwidth}
    \centering
    \includegraphics[width=\linewidth]{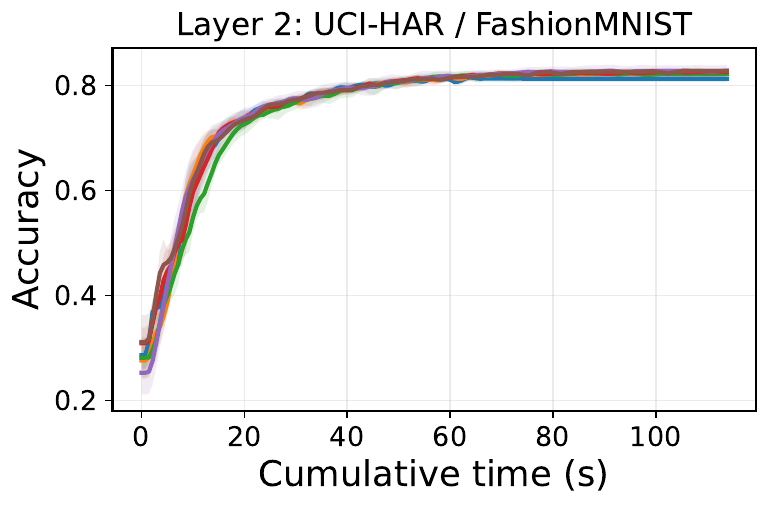}
    \caption{Layer 2.}
    \label{fig:layer2_1}
  \end{subfigure}
  \hfill
  \begin{subfigure}[t]{0.31\textwidth}
    \centering
    \includegraphics[width=\linewidth]{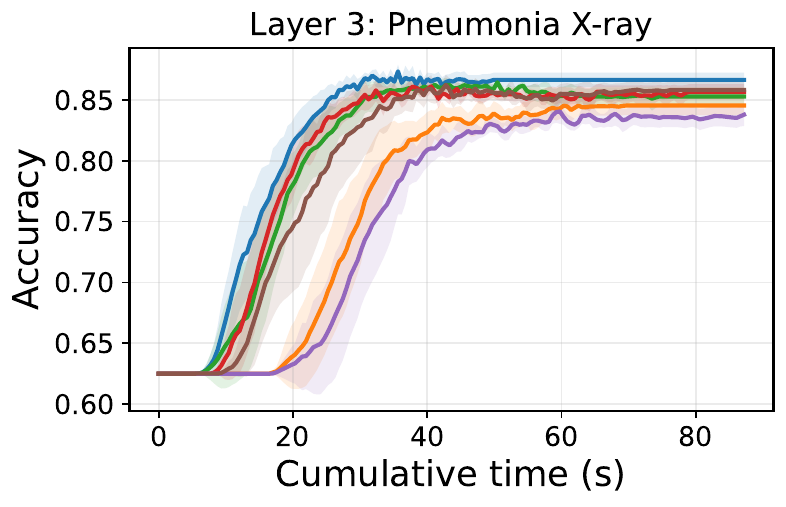}
    \caption{Layer 3.}
    \label{fig:layer3_1}
  \end{subfigure}

  \caption{Accuracy vs. time for different model layers. Solid lines report six-seed mean accuracy, and shaded bands denote 95\% confidence intervals.}
  \label{fig:3-100}
\end{figure*}

\subsection{Scalability Analysis}
To evaluate scalability with respect to the number of participating edge clients, we rerun the scalability study with $|\mathcal U|\in\{50,100,200,300,500\}$. The participation budget is scaled with the client pool size so that the selected-client fraction remains consistent across scales: MNIST uses a 50\% cap and UCI-HAR, FashionMNIST, and Pneumonia X-ray use 75\% caps. This avoids conflating scalability with an artificial reduction in participation rate at larger client counts. Because the middle-layer tasks are executed in parallel under the DAG, completion time is measured as the time for MNIST plus the slower of UCI-HAR and FashionMNIST plus Pneumonia X-ray, rather than as a simple sum over all tasks.

Fig.~\ref{fig:scalability} reports the target-reaching DAG completion time and total rounds for the scalability sweep. All methods reach the target thresholds at every client count over the six seeds. A-CoDa achieves the lowest mean DAG completion time at all client counts, with completion times of 147.4 s, 85.2 s, 51.9 s, 39.9 s, and 29.7 s for 50, 100, 200, 300, and 500 clients, respectively. Relative to the closest non-PSI baseline, the corresponding reductions are 9.7\%, 0.6\%, 12.3\%, 13.9\%, and 12.7\%. Overall, the scalability results confirm that the proposed dependency-aware client selection remains effective as the client pool grows, and that A-CoDa's latency-aware selection benefits remain most visible in the larger-client regimes.

\begin{figure*}[t]
  \centering
  \includegraphics[width=0.88\textwidth]{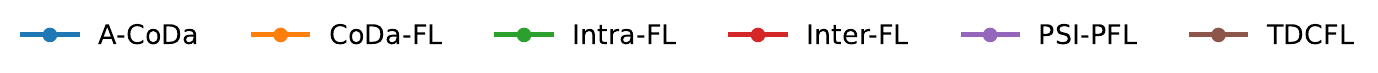}
  \vspace{0.15em}
  \par
  \begin{subfigure}[t]{0.46\textwidth}
    \centering
    \includegraphics[width=\linewidth]{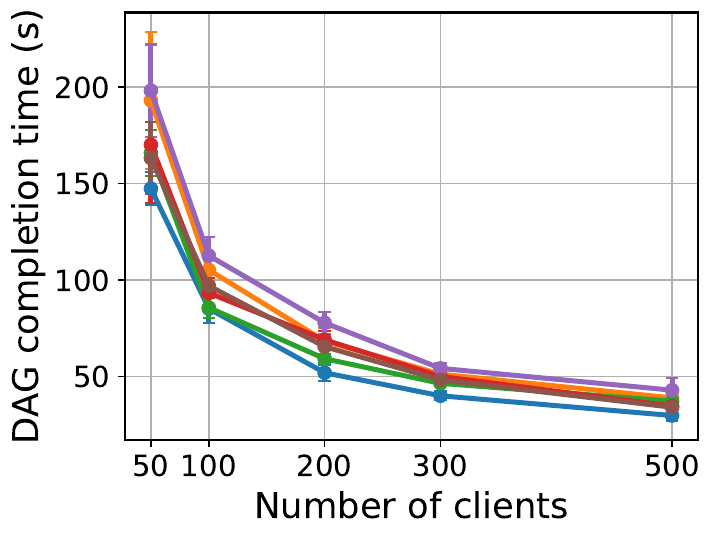}
    \caption{DAG completion time with different numbers of clients.}
    \label{fig:scalability-time}
  \end{subfigure}
  \begin{subfigure}[t]{0.46\textwidth}
    \centering
    \includegraphics[width=\linewidth]{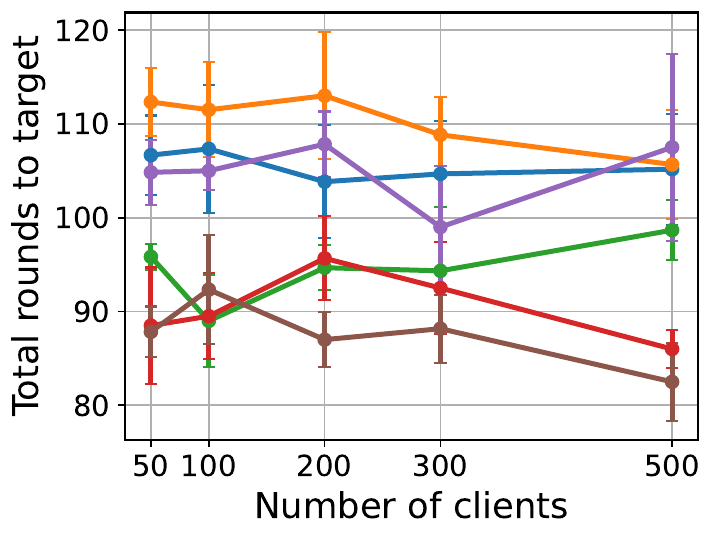}
    \caption{Total rounds to target with different numbers of clients.}
    \label{fig:scalability-rounds}
  \end{subfigure}
  \caption{Scalability with different numbers of clients. The participation cap is scaled with client count, and error bars denote 95\% confidence intervals over six seeds.}
  \label{fig:scalability}
\end{figure*}

\section{Conclusion}
In this paper, we proposed A-CoDa, a unified and dynamically adaptive framework for addressing client clustering and coordination challenges in latency-aware and dependency-aware federated learning workflows over heterogeneous edge networks. We first established a theoretical foundation that links convergence behavior to statistical heterogeneity, local-SGD drift, and loss-driven variance, revealing how these factors affect the sufficient round bound and the steady-state error floor. Guided by these insights, we developed an LDD-based balanced clustering scheme and FedMIX, a dynamic cluster participation mechanism that adapts to time-varying network conditions, client availability, and task dependencies. Experiments on handwriting, wearable-sensing, product-image, and medical-imaging tasks validated the effectiveness of A-CoDa, demonstrating reduced end-to-end completion time while maintaining competitive accuracy. Overall, this work provides a principled framework for scalable, dependency-aware, and latency-efficient federated learning workflow management in heterogeneous edge networks.


\appendices
\section{Proof of Proposition~\ref{prop:upper-steady}}
\label{app:proof-prop-upper-steady}

Let $e^{(r)}\triangleq b^{(r)}+d^{(r)}$ and abbreviate $\nabla F(w^{(r)})$ by $\nabla F$. By the descent lemma for the $L$-smooth objective $F$ and the update $w^{(r+1)}=w^{(r)}-\eta g_r$, we have
\begin{align}
\mathbb E[F(w^{(r+1)})\mid w^{(r)}]
&\le F(w^{(r)})
-\eta\langle \nabla F,\mathbb E[g_r\mid w^{(r)}]\rangle
+\frac{L\eta^2}{2}\mathbb E\|g_r\|^2 .
\label{eq:app-pre-young}
\end{align}
Using $\mathbb E[g_r\mid w^{(r)}]=\nabla F+e^{(r)}$ and~\eqref{eq:app-second-moment},
\begin{align}
\mathbb E[F(w^{(r+1)})\mid w^{(r)}]
&\le F(w^{(r)})-\eta\|\nabla F\|^2
-\eta\langle \nabla F,e^{(r)}\rangle
\notag\\
&\quad +\frac{L\eta^2}{2}\|\nabla F+e^{(r)}\|^2
+\frac{L\eta^2}{2}\sigma_r^2 .
\end{align}
Young's inequality gives
$-\langle\nabla F,e^{(r)}\rangle\le \frac14\|\nabla F\|^2+\|e^{(r)}\|^2$, and
$\|\nabla F+e^{(r)}\|^2\le2\|\nabla F\|^2+2\|e^{(r)}\|^2$. Therefore,
\begin{align}
\mathbb E[F(w^{(r+1)})\mid w^{(r)}]
&\le F(w^{(r)})
-\left(\frac{3\eta}{4}-L\eta^2\right)\|\nabla F\|^2
\notag\\
&\quad +(\eta+L\eta^2)\|e^{(r)}\|^2
+\frac{L\eta^2}{2}\sigma_r^2 .
\label{eq:app-progress-pre}
\end{align}
Under $0<\eta\le 1/(4L)$, we have $\frac{3\eta}{4}-L\eta^2\ge\eta/2$ and $\eta+L\eta^2\le2\eta$. Applying the PL condition and $\|e^{(r)}\|^2\le2\|b^{(r)}\|^2+2\|d^{(r)}\|^2$ gives
\begin{align}
\mathbb E[F(w^{(r+1)})-F(w^*)]
&\le
\left(1-\frac{\eta\mu}{2}\right)\mathbb E[F(w^{(r)})-F(w^*)]
\notag\\
&\quad +4\eta\mathbb E\left[\|b^{(r)}\|^2+\|d^{(r)}\|^2\right]
+\frac{L\eta^2}{2}\sigma_r^2 .
\label{eq:app-one-step-final}
\end{align}
By Assumption~\ref{ass2} and $\sum_u(q_r(u)-\pi_u)=0$,
\begin{align}
\|b^{(r)}\|
&=\left\|\sum_u(q_r(u)-\pi_u)\big(\nabla F_u(w^{(r)})-\nabla F(w^{(r)})\big)\right\|
\notag\\
&\le C_g\sum_u|q_r(u)-\pi_u|\Delta_u
+\sum_u|q_r(u)-\pi_u|\chi_u(w^{(r)})
\notag\\
&\le C_g\bar B+\bar\chi .
\label{eq:app-bias}
\end{align}
Assumption~\ref{ass4} gives $\|d^{(r)}\|^2\le\bar D_E$, and~\eqref{eq:signma} gives
\begin{equation}
\sigma_r^2\le M\bar G_{\mathrm{agg}}^2+\bar\Sigma.
\label{eq:app-sigma}
\end{equation}
Taking total expectation in~\eqref{eq:app-one-step-final} and substituting~\eqref{eq:app-bias}--\eqref{eq:app-sigma}, we obtain
\begin{equation}
X_{r+1}
\le
\left(1-\frac{\eta\mu}{2}\right)X_r+\Psi,
\label{eq:app-master}
\end{equation}
where $X_r=\mathbb E[F(w^{(r)})-F(w^*)]$ and
\begin{equation}
\Psi=
4\eta\left((C_g\bar B+\bar\chi)^2+\bar D_E\right)
+\frac{L\eta^2}{2}\left(M\bar G_{\mathrm{agg}}^2+\bar\Sigma\right).
\end{equation}
Let $\varrho=\eta\mu/2$. Unrolling~\eqref{eq:app-master} gives
\begin{align}
X_R
&\le (1-\varrho)^RX_0
+\frac{1-(1-\varrho)^R}{\varrho}\Psi
\notag\\
&\le (1-\varrho)^RX_0+\frac{\Psi}{\varrho}.
\label{eq:app-unroll}
\end{align}
Since $E_\infty=\Psi/\varrho=2\Psi/(\eta\mu)$, this proves~\eqref{eq:finite-time}. Taking $R\to\infty$ yields~\eqref{eq:steady-state}. Finally, solving $(1-\varrho)^RX_0\le\varepsilon-E_\infty$ and using $(1-\varrho)^R\le e^{-\varrho R}$ gives the sufficient round bound~\eqref{eq:rounds}. This completes the proof.

\bibliographystyle{IEEEtran}
\bibliography{ref}

@IEEEtranBSTCTL{IEEEexample:BSTcontrol,
  CTLuse_article_number = "yes",
  CTLuse_paper = "yes",
  CTLuse_url = "yes",
  CTLuse_forced_etal = "no",
  CTLmax_names_forced_etal = "10",
  CTLnames_show_etal = "1"
}

@article{jimenez2025psi,
  title={{PSI-PFL}: Population Stability Index for Client Selection in {{N}on-{IID}} Personalized Federated Learning},
  author={Jimenez-Gutierrez, Daniel-M and Solans, David and Elbamby, Mohammed and Kourtellis, Nicolas},
  journal={arXiv preprint arXiv:2506.00440},
  year={2025}
}

@article{ren2025dynamic,
  title={Dynamic clustered federated learning via adaptive distribution similarity computation},
  author={Ren, Tian and Cheng, Siyao and Zhang, Hao and Liu, Jie},
  journal={Comput. Netw.},
  pages={111302},
  year={2025},
  publisher={Elsevier}
}

@INPROCEEDINGS{8761315,
  author={Nishio, Takayuki and Yonetani, Ryo},
  booktitle={Proc. IEEE Int. Conf. Commun.}, 
  title={Client Selection for Federated Learning with Heterogeneous Resources in Mobile Edge}, 
  year={2019},
  volume={},
  number={},
  pages={1-7},
  doi={10.1109/ICC.2019.8761315}}

@inproceedings{cho2022towards,
  title={Towards understanding biased client selection in federated learning},
  author={Cho, Yae Jee and Wang, Jianyu and Joshi, Gauri},
  booktitle={Proc. Int. Conf. Artif. Intell. Stat.},
  pages={10351--10375},
  year={2022},
  organization={PMLR}
}

@inproceedings{wang2025fedccs,
  title={Fed{C}cs: Federated learning cluster-based client delection algorithm for {N}on-{IID} data},
  author={Wang, Jingyuan},
  booktitle={Int. Conf. Neural Netw., Inf. Commun. Eng. (NNICE)},
  pages={135--139},
  year={2025},
  organization={IEEE}
}

@ARTICLE{9846900,
  author={Shi, Fang and Hu, Chunchao and Lin, Weiwei and Fan, Lisheng and Huang, Tiansheng and Wu, Wentai},
  journal={IEEE Internet Things J.}, 
  title={{VF}ed{CS}: Optimizing Client Selection for Volatile Federated Learning}, 
  year={2022},
  volume={9},
  number={24},
  pages={24995-25010},
  doi={10.1109/JIOT.2022.3195073}}

@INPROCEEDINGS{9443523,
  author={Jee Cho, Yae and Gupta, Samarth and Joshi, Gauri and Yağan, Osman},
  booktitle={Asilomar Conf. Signals, Syst., Comput.}, 
  title={Bandit-based communication-efficient client selection strategies for federated learning}, 
  year={2020},
  volume={},
  number={},
  pages={1066-1069},
  doi={10.1109/IEEECONF51394.2020.9443523}}

@INPROCEEDINGS{10074237,
  author={Arisdakessian, Sarhad and Wahab, Omar Abdel and Mourad, Azzam and Otrok, Hadi},
  booktitle={Int. Conf. Comput., Netw., Commun.}, 
  title={Towards Instant Clustering Approach for Federated Learning Client Selection}, 
  year={2023},
  volume={},
  number={},
  pages={409-413},
  doi={10.1109/ICNC57223.2023.10074237}}

@ARTICLE{10197242,
  author={Huang, Honglan and Shi, Wei and Feng, Yanghe and Niu, Chaoyue and Cheng, Guangquan and Huang, Jincai and Liu, Zhong},
  journal={IEEE Trans. Neural Netw. Learn. Syst.}, 
  title={Active Client Selection for Clustered Federated Learning}, 
  year={2024},
  volume={35},
  number={11},
  pages={16424-16438},
  doi={10.1109/TNNLS.2023.3294295}}

@article{datta2024blockchain,
  title={Blockchain-based smart contract model for securing healthcare transactions by using consumer electronics and mobile-edge computing},
  author={Datta, Sagnik and Namasudra, Suyel},
  journal={IEEE Trans. Consum. Electron.},
  volume={70},
  number={1},
  pages={4026--4036},
  year={2024},
  publisher={IEEE}
}

@ARTICLE{10195234,
  author={He, Yejun and Yang, Mengna and He, Zhou and Guizani, Mohsen},
  journal={IEEE Trans. Cogn. Commun. Netw.}, 
  title={Computation Offloading and Resource Allocation Based on {DT-MEC}-Assisted Federated Learning Framework}, 
  year={2023},
  volume={9},
  number={6},
  pages={1707-1720},
  doi={10.1109/TCCN.2023.3298926}}

@INPROCEEDINGS{9685698,
  author={Cheng, Zhipeng and Min, Minghui and Liwang, Minghui and Gao, Zhibin and Huang, Lianfen},
  booktitle={2021 IEEE Global Commun. Conf.}, 
  title={Joint Client Selection and Task Assignment for Multi-Task Federated Learning in {MEC} Networks}, 
  year={2021},
  volume={},
  number={},
  pages={1-6},
  doi={10.1109/GLOBECOM46510.2021.9685698}}

@INPROCEEDINGS{9498853,
  author={Zheng, Jingjing and Li, Kai and Tovar, Eduardo and Guizani, Mohsen},
  booktitle={2021 International Wirel. Commun. Mob. Com. (IWCMC)}, 
  title={Federated Learning for Energy-balanced Client Selection in Mobile Edge Computing}, 
  year={2021},
  volume={},
  number={},
  pages={1942-1947},
  doi={10.1109/IWCMC51323.2021.9498853}}

@ARTICLE{10713971,
  author={Tong, Zhao and Deng, Jiaxin and Mei, Jing and Zhang, Yuanyang and Li, Keqin},
  journal={IEEE Trans. Serv. Comput.}, 
  title={Multi-Objective {DAG} Task Offloading in {MEC} Environment Based on Federated {DQN} With Automated Hyperparameter Optimization}, 
  year={2024},
  volume={17},
  number={6},
  pages={3999-4012},
  doi={10.1109/TSC.2024.3478841}}

@ARTICLE{11223124,
  author={Gu, Ke and Lei, Jiaqi and Tan, Jingjing and Li, Xiong},
  journal={IEEE Transactions on Network and Service Management}, 
  title={A Verifiable Federated Learning Scheme With Privacy-Preserving in MCS}, 
  year={2026},
  volume={23},
  number={},
  pages={862-879},
  doi={10.1109/TNSM.2025.3627581}}

@ARTICLE{10206024,
  author={Li, Hui and Li, Xiuhua and Fan, Qilin and Xiong, Qingyu and Wang, Xiaofei and Leung, Victor C. M.},
  journal={IEEE Transactions on Network and Service Management}, 
  title={Transfer Learning for Real-Time Surface Defect Detection With Multi-Access Edge-Cloud Computing Networks}, 
  year={2024},
  volume={21},
  number={1},
  pages={310-323},
  doi={10.1109/TNSM.2023.3301718}}

@ARTICLE{qi2025energy,
  author={Qi, Hang and Luo, Jieping and Xu, Zimu and Li, Qiyue and Yin, Jiaying and Wu, Jingjin},
  journal={IEEE Commun. Lett.}, 
  title={Energy Efficient Power Control for Over-the-Air Federated Learning in Satellite Communications}, 
  year={2025},
  volume={29},
  number={7},
  pages={1530-1534}
}

@article{qi2025comparative,
  author    = {Hang Qi and Jieping Luo and Qiyue Li and Jingjin Wu},
  title     = {A comparative trade-off analysis on accuracy and efficiency for federated learning in demand forecasting},
  journal   = {Appl. Soft Comput.},
  volume    = {182},
  year      = {2025},
  pages     = {113561},
  issn      = {1568-4946},
  doi       = {10.1016/j.asoc.2025.113561},
}

@article{yang2019federated,
  title={Federated machine learning: Concept and applications},
  author={Yang, Qiang and Liu, Yang and Chen, Tianjian and Tong, Yongxin},
  journal={ACM Trans. Intell. Syst. Technol.},
  volume={10},
  number={2},
  pages={1--19},
  year={2019},
  publisher={ACM New York, NY, USA}
}

@inproceedings{lai2021oort,
  title={Oort: Efficient federated learning via guided participant selection},
  author={Lai, Fan and Zhu, Xiangfeng and Madhyastha, Harsha V and Chowdhury, Mosharaf},
  booktitle={15th USENIX Symp. Oper. Syst. Des. Implement.},
  pages={19--35},
  year={2021}
}

@ARTICLE{10468591,
  author={Lu, Zili and Pan, Heng and Dai, Yueyue and Si, Xueming and Zhang, Yan},
  journal={IEEE Internet Things J.}, 
  title={Federated Learning With {N}on-{IID} Data: A Survey}, 
  year={2024},
  volume={11},
  number={11},
  pages={19188-19209},
  doi={10.1109/JIOT.2024.3376548}}

@article{liu2023dependent,
  title={Dependent task scheduling and offloading for minimizing deadline violation ratio in mobile edge computing networks},
  author={Liu, Shumei and Yu, Yao and Lian, Xiao and Feng, Yuze and She, Changyang and Yeoh, Phee Lep and Guo, Lei and Vucetic, Branka and Li, Yonghui},
  journal={IEEE J. Sel. Areas Commun. },
  volume={41},
  number={2},
  pages={538--554},
  year={2023},
  publisher={IEEE}
}

@inproceedings{zhou2022efficient,
  title={Efficient device scheduling with multi-job federated learning},
  author={Zhou, Chendi and Liu, Ji and Jia, Juncheng and Zhou, Jingbo and Zhou, Yang and Dai, Huaiyu and Dou, Dejing},
  booktitle={Proc. of the AAAI Conf. Artif. Intell.},
  volume={36},
  number={9},
  pages={9971--9979},
  year={2022}
}

@article{rubner2000earth,
  title={The {Earth Mover's Distance} as a metric for image retrieval},
  author={Rubner, Yossi and Tomasi, Carlo and Guibas, Leonidas J},
  journal={Int. J. Comput. Vis},
  volume={40},
  pages={99--121},
  year={2000},
  publisher={Springer}
}

@ARTICLE{zhao2021ServiceCaching,
  author={Zhao, Gongming and Xu, Hongli and Zhao, Yangming and Qiao, Chunming and Huang, Liusheng},
  journal={IEEE Trans. Parallel Distrib. Syst.}, 
  title={Offloading Tasks With Dependency and Service Caching in Mobile Edge Computing}, 
  year={2021},
  volume={32},
  number={11},
  pages={2777-2792},
 }

@ARTICLE{dai2023offloading,
  author={Dai, Xingxia and Xiao, Zhu and Jiang, Hongbo and Lei, Ming and Min, Geyong and Liu, Jiangchuan and Dustdar, Schahram},
  journal={IEEE Trans. Serv. Comput.}, 
  title={Offloading Dependent Tasks in Edge Computing With Unknown System-Side Information}, 
  year={2023},
  volume={16},
  number={6},
  pages={4345-4359},
  }

@inproceedings{chai2020tifl,
  title={Tifl: A tier-based federated learning system},
  author={Chai, Zheng and Ali, Ahsan and Zawad, Syed and Truex, Stacey and Anwar, Ali and Baracaldo, Nathalie and Zhou, Yi and Ludwig, Heiko and Yan, Feng and Cheng, Yue},
  booktitle={Proc. 29th Int. Symp. High-Perform. Parallel Distrib. Comput.},
  pages={125--136},
  year={2020}
}

@inproceedings{li2024adafl,
  title={Adafl: Adaptive client selection and dynamic contribution evaluation for efficient federated learning},
  author={Li, Qingming and Li, Xiaohang and Zhou, Li and Yan, Xiaoran},
  booktitle={2024 IEEE Int. Conf. Acoust., Speech, Signal Process.},
  pages={6645--6649},
  year={2024},
  organization={IEEE}
}

@article{zhang2022deepemd,
  title={Deepemd: Differentiable {Earth Mover's Distance} for few-shot learning},
  author={Zhang, Chi and Cai, Yujun and Lin, Guosheng and Shen, Chunhua},
  journal={IEEE Trans. Pattern Anal. Mach. Intell.},
  volume={45},
  number={5},
  pages={5632--5648},
  year={2022},
  publisher={IEEE}
}

@inproceedings{luo2025cluster,
  title={Cluster-Based Client Selection for Dependent Multi-Task Federated Learning in Edge Computing},
  author={Luo, Jieping and Li, Qiyue and Liu, Zhizhang and Qi, Hang and Yin, Jiaying and Wu, Jingjin},
  booktitle={Proc. IEEE GlobeCom 2025 Wkshps},
  year={2025}
}

@ARTICLE{10018536,
  author={Zhang, Yu and Liu, Duo and Duan, Moming and Li, Li and Chen, Xianzhang and Ren, Ao and Tan, Yujuan and Wang, Chengliang},
  journal={IEEE Trans. Parallel Distrib. Syst.}, 
  title={Fed{MDS}: An Efficient Model Discrepancy-Aware Semi-Asynchronous Clustered Federated Learning Framework}, 
  year={2023},
  volume={34},
  number={3},
  pages={1007-1019},
  doi={10.1109/TPDS.2023.3237752}}

@inproceedings{alekseenko2024distance,
  title={Distance-aware {N}on-{IID} federated learning for generalization and personalization in medical imaging segmentation},
  author={Alekseenko, Julia and Karargyris, Alexandros and Padoy, Nicolas},
  booktitle={Med. Imaging Deep Learn.},
  year={2024}
}

@article{chen2022emd,
  title={An emd-based adaptive client selection algorithm for federated learning in heterogeneous data scenarios},
  author={Chen, Aiguo and Fu, Yang and Sha, Zexin and Lu, Guoming},
  journal={Front. Plant. Sci.},
  volume={13},
  pages={908814},
  year={2022},
  publisher={Frontiers Media SA}
}

@inproceedings{shen2018wasserstein,
  title={Wasserstein distance guided representation learning for domain adaptation},
  author={Shen, Jian and Qu, Yanru and Zhang, Weinan and Yu, Yong},
  booktitle={Proc. AAAI Conf. Artif. Intell.},
  volume={32},
  number={1},
  year={2018}
}

@INPROCEEDINGS{9797945,
  author={Lin, Daoqin and Guo, Yuchun and Sun, Huan and Chen, Yishuai},
  booktitle={Proc. IEEE INFOCOM - IEEE Conf. Comput. Commun. Workshops}, 
  title={FedCluster: A Federated Learning Framework for Cross-Device Private {ECG} Classification}, 
  year={2022},
  volume={},
  number={},
  pages={1-6},
  doi={10.1109/INFOCOMWKSHPS54753.2022.9797945}}

@article{luo2023influence,
  title={Influence of data distribution on federated learning performance in tumor segmentation},
  author={Luo, Guibo and Liu, Tianyu and Lu, Jinghui and Chen, Xin and Yu, Lequan and Wu, Jian and Chen, Danny Z and Cai, Wenli},
  journal={Radiol. Artif. Intell.},
  volume={5},
  number={3},
  pages={e220082},
  year={2023},
  publisher={Radiological Society of North America}
}

@inproceedings{anguita2013public,
  title={A Public Domain Dataset for Human Activity Recognition Using Smartphones},
  author={Anguita, Davide and Ghio, Alessandro and Oneto, Luca and Parra, Xavier and Reyes-Ortiz, Jorge L.},
  booktitle={Proc. 21st Eur. Symp. Artificial Neural Networks, Computational Intelligence and Machine Learning},
  pages={437--442},
  year={2013}
}

@article{yang2023medmnist,
  title={{MedMNIST} v2: A Large-Scale Lightweight Benchmark for 2D and 3D Biomedical Image Classification},
  author={Yang, Jiancheng and Shi, Rui and Wei, Donglai and Liu, Zequan and Zhao, Lin and Ke, Bilian and Pfister, Hanspeter and Ni, Bingbing},
  journal={Scientific Data},
  volume={10},
  number={1},
  pages={41},
  year={2023},
  publisher={Nature Publishing Group},
  doi={10.1038/s41597-022-01721-8}
}

@article{kermany2018identifying,
  title={Identifying Medical Diagnoses and Treatable Diseases by Image-Based Deep Learning},
  author={Kermany, Daniel S. and others},
  journal={Cell},
  volume={172},
  number={5},
  pages={1122--1131.e9},
  year={2018},
  doi={10.1016/j.cell.2018.02.010}
}

\end{document}